\documentclass[twocolumn,10pt,journal]{IEEEtran}
\usepackage{color}
\usepackage{indentfirst}
\usepackage{amsfonts,amsmath,amsthm,amssymb}
\usepackage{graphicx,multirow,bm}
\usepackage{cite}

\newtheorem{Construction}{Construction}
\newtheorem{Definition}{Definition}
\newtheorem{Lemma}{Lemma}
\newtheorem{Theorem}{Theorem}
\newtheorem{Remark}{Remark}

\newenvironment{psmallmatrix}
  {\left(\begin{smallmatrix}}
  {\end{smallmatrix}\right)}

\begin{document}

\title{PMDS Array Codes With Small Sub-packetization,  Small Repair Bandwidth/Rebuilding Access
}

\author{Jie~Li,~\IEEEmembership{Member,~IEEE,}
        Xiaohu~Tang,~\IEEEmembership{Senior~Member,~IEEE,} Hanxu~Hou,~\IEEEmembership{Member,~IEEE,}\\
        Yunghsiang S.~Han,~\IEEEmembership{Fellow,~IEEE,} Bo~Bai,~\IEEEmembership{Senior~Member,~IEEE,} and Gong Zhang
\thanks{
The work of X. Tang was supported in part by the National Natural Science Foundation of China under Grant 61871331. This paper was presented in part at the 2022 IEEE
International Symposium on Information Theory \cite{li2022pmds}.}
\thanks{J. Li, H. Hou, B. Bai, and G. Zhang  are with the Theory Lab, Central Research Institute, 2012 Labs, Huawei Technologies Co., Ltd., Hong Kong SAR, China (e-mails: li.jie9@huawei.com; hou.hanxu@huawei.com; baibo8@huawei.com; nicholas.zhang@huawei.com).}
\thanks{X. Tang is Information Coding and Transmission Key Lab of Sichuan Province, CSNMT Int. Coop. Res. Centre (MoST), Southwest Jiaotong University, Chengdu, 610031, China (e-mail: xhutang@swjtu.edu.cn).}
\thanks{Yunghsiang S.~Han is with the Shenzhen Institute for Advanced Study, University of Electronic Science and Technology of China, Shenzhen, 518110, China (e-mail: yunghsiangh@gmail.com).}
}

\maketitle

\begin{abstract}
Partial maximum distance separable (PMDS) codes are a kind of erasure codes where the nodes are divided into multiple groups with each forming an MDS code with a smaller code length, thus they allow repairing a failed node with only a few helper nodes  and can correct all erasure patterns that
are information-theoretically correctable.
However, the repair of a failed node of PMDS codes still requires a large amount of communication if the group size is large.  Recently, PMDS array codes with each local code being an MSR code were introduced to reduce the repair bandwidth further. However, they require extensive rebuilding access and unavoidably a significant sub-packetization level. In this paper, we first propose two constructions of PMDS array codes with two global parities that have smaller sub-packetization levels and much smaller finite fields than the existing one. One construction can support an arbitrary number of local parities and has $(1+\epsilon)$-optimal repair bandwidth (i.e., $(1+\epsilon)$ times the optimal repair bandwidth), while the other one is limited to two local parities but has significantly smaller rebuilding access and its sub-packetization level is only $2$. In addition, we present a construction of PMDS array code with three global parities, which has a smaller sub-packetization level as well as $(1+\epsilon)$-optimal repair bandwidth, the required finite field is significantly smaller than existing ones.
\end{abstract}


\section{Introduction}
\IEEEPARstart{W}ith the rapid increase in the volumes of data stored online, traditional storage techniques such as duplicating or triplicating data are not economically feasible. This has resulted in erasure coding-based distributed storage systems, which can provide reliability with low storage overhead. Previous distributed storage systems usually call upon the maximum distance separable (MDS) codes, which provide the optimal tradeoff
between fault tolerance and storage overhead. However, the downside of employing MDS codes is the excessive repair bandwidth and rebuilding access when repairing a failed node, where \textit{repair bandwidth} is defined as the amount of data downloaded from helper nodes to repair a failed node and \textit{rebuilding access} is defined as the amount of data accessed. Consider a distributed storage system that is based on an $[n,k]$ MDS code, repairing a failed node requires accessing and downloading the entire content from any $k$ surviving nodes, thus leading to a large amount of access and bandwidth.

To reduce the repair bandwidth, regenerating codes were introduced in the pioneering work \cite{dimakis2010network}, which allow for repairing a failed node by contacting more than $k$ surviving nodes but only downloading a fraction of the data stored at each node. The optimal tradeoff between the storage and repair bandwidth was also characterized in \cite{dimakis2010network}, which leads to two extremal classes of codes, namely  \emph{minimum storage regenerating (MSR)} codes and \emph{minimum bandwidth regenerating (MBR)} codes. MSR codes are a kind of MDS codes as they are optimal in terms of storage overhead, whereas MBR codes result in more storage overhead but can offer the minimum repair bandwidth. Regenerating codes including the MSR codes and MBR codes have attracted a lot of attention in the past decade \cite{rashmi2011optimal,tamo2012zigzag,papailiopoulos2013repair,li2018generic,li2016optimal,li2017generic,wang2016explicit,wang2011codes,li2015framework,ye2017explicit,ye2017explicitB,sasidharan2016explicit,goparaju2017minimum,elyasi2020cascade,balaji2018erasure,tian2014characterizing,hou2019binary,han2015update,kralevska2017hashtag,chen2020explicit}.

Although regenerating codes can significantly reduce the repair bandwidth, however, a large number of helper nodes are required to be contacted when repairing a failed node. As an alternative and parallel coding technique, locally repairable codes (LRCs) require only a few helper nodes during the repair process, however, at the cost of introducing additional redundancy to the system. Studies on the upper bound of the minimum distance of LRCs as well as the optimal constructions have also attracted a lot of attention in the past decade \cite{gopalan2012locality,kamath2014codes,rawat2014optimal,tamo2014family,wang2014repair,luo2018optimal,guruswami2019long,chen2020improved,cai2020optimal,cai2021optimal}.
In another line of research, Partial MDS (PMDS) codes \cite{blaum2013partial} provide an alternative solution, which are a strictly stronger class of LRCs as they are not only distance-optimal LRCs but can also correct any erasure pattern that is information-theoretically correctable. Formally, a $(\mu,n;r,s)$ PMDS code is a $[\mu n, \mu (n-r)-s]$ linear code, which can be partitioned into $\mu$ groups each of size $n$, such that any erasure pattern with $r$ erasures in each group plus any $s$ erasures in arbitrary places can be tolerated. Note that maximally recoverable codes are also referred to as PMDS codes \cite{chen2007maximally,gopalan2014explicit} when restricted to the RAID-type architecture. Besides, they can be applied to more topologies, e.g., see \cite{gopalan2017maximally} for grid-like topologies and \cite{shivakrishna2018maximally} for product topologies.

\begin{table*}[htbp]
\begin{center}
\caption{Comparison of key parameters among new $(\mu, n; r, s=2)$ PMDS array codes and the one in \cite{holzbaur2021partial}, where $n'\ge 2$ and we assume  $n'|n$ for convenience of notation}\label{Ta:con_r=2}
\setlength{\tabcolsep}{3.7pt}
\begin{tabular}{|c|c|c|c|c|c|c|}
  \hline
&   $r$&  $s$ & Sub-packetization level $\ell$ & Field size $q$ & Repair bandwidth $\gamma$ & Rebuilding access $\Gamma $  \\
 \hline  Construction A in \cite{holzbaur2021partial}
  &$\ge2 $ &$2$ & $r^{n}$ & $>\mu r(rn-r+n-2)$ & $\frac{\ell}{r}(n-1)$ & $\ell(n-1)$\\
  \hline Construction \ref{Con-new1} (Thm. \ref{Eqn_Thm_C1})
  &$\ge2 $ &$2$ & $r^{n'}$ & $>\mu rn'\lceil\frac{n}{rn'}\rceil$ & $(1+\frac{(\frac{n}{n'}-1)(r-1)}{n-1})\frac{\ell}{r}(n-1)$ & $\ell(n-1)$\\
\hline
     Construction \ref{Con-PMDS2} (Thms. \ref{Thm-band-local} and \ref{Thm-Con3})
  &$2 $ &$2$ & $2$ & $>\mu n$ & $\frac{3n}{2}-2$ & $\frac{3n}{2}-2$\\
  \hline
\end{tabular}
\end{center}
\end{table*}

\begin{table*}[htbp]
\begin{center}
\caption{Comparison of key parameters among new $(\mu, n; r, s=3)$ PMDS array codes and existing ones in \cite{holzbaur2021partial}, where we assume  $n'|n$ for convenience of notation}\label{Ta:conC1}
\begin{tabular}{|c|c|c|c|c|c|c|}
  \hline
&   $r$&  $s$ & Sub-packetization $\ell$ & Field size $q$ & Repair bandwidth $\gamma$ & Rebuilding access $\Gamma $  \\
\hline \hline Construction B in \cite{holzbaur2021partial}
  &$\ge2 $ &$3$ & $r^{n}$ & $\ge(rn)^{\mu(n-r)}$ & $\frac{\ell}{r}(n-1)$ & $\ell(n-1)$\\
\hline  Construction C in \cite{holzbaur2021partial}
  &$\ge2 $ &$3$ & $r^{n}$ & $\ge\max\{rn, \mu+1\}^{n-r}$ & $\frac{\ell}{r}(n-1)$ & $\ell(n-1)$\\
\hline  Construction D in \cite{holzbaur2021partial}
  &$\ge2 $ &$3$ & $r^{n}$ & $\ge rn(\mu n)^{3(r+1)-1}$ & $\frac{\ell}{r}(n-1)$ & $\ell(n-1)$\\
  \hline Construction \ref{Con-news3} (Thm. \ref{Eqn_Thm_C3})
  &$\ge2 $ &$3$ & $r^{n'}$ & $>(\mu rn'\lceil\frac{n}{rn'}\rceil+1)^3$ & $(1+\frac{(\frac{n}{n'}-1)(r-1)}{n-1})\frac{\ell}{r}(n-1)$ & $\ell(n-1)$\\
  \hline
\end{tabular}
\end{center}
\end{table*}

In general, PMDS codes are much harder to obtain than LRC codes, especially over relatively small finite fields. Existing results \cite{blaum2013partial,gopalan2014explicit,blaum2016construction,calis2016general,gabrys2019constructions,gopi2020maximally,cai2021construction,guruswami2020constructions} show that a finite field with an exponential size is required when $r\ge 2$ and the number $s$ of global parities exceeds $3$. Note that repairing a failed node of PMDS codes may still require a large amount of communication if the group size is large, as the whole content of the surviving nodes in the same group needs to be downloaded. LRCs also have the same issue. In \cite{kamath2014codes,rawat2014optimal}, the idea of using regenerating codes (including MSR codes and MBR codes) to encode the local group was first introduced to LRCs for reducing the repair bandwidth further.
Very recently, by using MSR codes to encode the local group, PMDS array codes were introduced in \cite{holzbaur2021partial}, which combines both the advantages of PMDS codes and MSR codes, i.e., 1) Can correct any erasure pattern that is information-theoretically correctable, 2) Require fewer helper nodes during the node repair process, and 3) Can further reduce the repair bandwidth. However, the explicit PMDS code constructions with a linear field size were only provided for two global parities, the constructions of PMDS array codes for more than two global parities still require a relatively large finite field.
Although the PMDS array codes in \cite{holzbaur2021partial} have the optimal repair bandwidth as each local group is encoded by an MSR code, they require large rebuilding access and unavoidably a large sub-packetization level.

In this paper, we focus on PMDS array codes with smaller sub-packetization levels and smaller field sizes than existing ones while endowing $(1+\epsilon)$-optimal repair bandwidth. More specifically, we restrict to PMDS array codes with two and three global parities, as it is promising to construct PMDS array codes with few global parities over small finite fields. We present two PMDS array codes with $s=2$ global parities and one with $s=3$ global parities. Comparisons of the key parameters among the newly proposed PMDS array codes and some existing ones are given in Tables \ref{Ta:con_r=2} and \ref{Ta:conC1} under $s=2$ and $s=3$, respectively.

From Tables \ref{Ta:con_r=2} and \ref{Ta:conC1}, we see that the new proposed PMDS array codes have the following advantages:
\begin{itemize}
\item 
The required finite fields of the two new PMDS array codes in Constructions \ref{Con-new1} and \ref{Con-PMDS2} are only a fraction of around $\frac{1}{r}$ to $\frac{1}{r^2}$  as Construction A in \cite{holzbaur2021partial}, which is based on the first MSR code construction in \cite{ye2017explicit} (i.e., each sub-stripe can be viewed as a scalar MDS code and one stripe contains $\ell$ sub-stripes) and the PMDS codes in \cite{blaum2016construction}.

\item The new PMDS array code in Construction \ref{Con-PMDS2}, which supports only two local parity nodes, has smaller rebuilding access (when normalized by the file size) and a smaller sub-packetization level compared with the code in Construction A in \cite{holzbaur2021partial} and the new code in Construction \ref{Con-new1}.

\item The required field size of the new code in Construction \ref{Con-news3} is significantly smaller than those in Constructions B-D in \cite{holzbaur2021partial} for almost all parameter ranges,
where Constructions B, C, and D in \cite{holzbaur2021partial} are obtained by combining a universal PMDS code and an MSR code with each sub-stripe being a scalar MDS code (e.g., the first MSR code construction in \cite{ye2017explicit}). More specifically, they employ the Gabidulin-code-based PMDS code in \cite{rawat2014optimal}, the linearized-RS-codes-based PMDS code in \cite{martinez2019universal}, and a generalization of the PMDS code in \cite{gabrys2019constructions} as the universal PMDS code, respectively.

\item The new codes in Constructions \ref{Con-new1} and \ref{Con-news3} can provide a flexible tradeoff between the sub-packetization level and the repair bandwidth by varying $n'$ in $[r+1, n+1)$ when the MDS array code in \cite{li2021systematic} with this flexible tradeoff is employed as the local code. The two extreme points of the tradeoff are $(\ell, \gamma)=\left(r^{r+1}, (1+\frac{(\frac{n}{r+1}-1)(r-1)}{n-1})\frac{\ell}{r}(n-1)\right)$ and $(\ell, \gamma)=\left(r^{n}, \frac{\ell}{r}(n-1)\right)$, where the latter one is also achieved by the constructions in \cite{holzbaur2021partial}. As an example, when $n=30$ and $r=2$, some of the sub-packetization levels and repair bandwidths that the new codes in Constructions \ref{Con-new1} and \ref{Con-news3} can provide are 
\begin{align*}
(\ell,\gamma)=&\left(2^3,(1+\frac{9}{29})\gamma^*\right),\left(2^5,(1+\frac{5}{29})\gamma^*\right),\\
& \left(2^6, (1+\frac{4}{29})\gamma^*\right), \left(2^{10}, (1+\frac{2}{29}\gamma^*)\right),\\
& \left(2^{30}, \gamma^*\right), 
\end{align*}
where all points except for the last one are new that can be achieved by Constructions \ref{Con-new1} and \ref{Con-news3}, and $\gamma^*=\frac{\ell}{r}(n-1)$ denotes the repair bandwidth of an $[n, n-r]$ MSR code with sub-packetization level $\ell$.

 \item When $n=n'$, i.e.,  each local group of the new Constructions \ref{Con-new1} and \ref{Con-news3} forms an MSR code, the sub-packetization level $\ell$ and repair bandwidth $\gamma$ are the same as those of Constructions B-D in \cite{holzbaur2021partial}. Nevertheless, the required field sizes of the new code Constructions \ref{Con-new1} and \ref{Con-news3} are still smaller than the PMDS array code Construction A in \cite{holzbaur2021partial} and Constructions B-D in \cite{holzbaur2021partial}, respectively, while all the other properties are the same.
\end{itemize}
Although Constructions B-D in \cite{holzbaur2021partial} require a huge finite field, it is worth pointing out that they work for all $s\ge 1$.

The rest of this paper is organized as follows. Section \ref{sec:pre} introduces some necessary preliminaries. Section \ref{sec:codes_r=2} proposes two new PMDS array code constructions with two global parities and $(1+\epsilon)$-optimal repair bandwidth. A new PMDS array code construction with three global parities and $(1+\epsilon)$-optimal repair bandwidth is presented in Section \ref{sec:code_r=3}. Finally, Section \ref{sec:Conclu} draws the conclusion.

\section{Preliminaries}\label{sec:pre}
First of all, we fix some notations used in this paper. Let $q$ be a prime power and $\mathbf{F}_q$ the finite field containing $q$ elements. For two integers $a$ and $b$ with $a<b$, denote by $[a, b)$ the set $\{a, a+1, \ldots, b-1\}$.

Let $\ell\ge 1$, for an $a\ell \times b\ell$ matrix  $A$, let $A(j)$ denote the $j$-th column of $A$, and $A(J)$ denote the sub-matrix of $A$  that formed by the columns of $A$ with indices in the set $J$ where $J\subset [0, b\ell)$. If $J=[j \ell, j \ell+1, \ldots, j \ell+\ell)$ for some $j\in [0, b)$, then we say that the sub-matrix $A(J)$ is the $j$-th \textit{thick} column of $A$ and denote it by $A[j]$. Let $A[J]$ denote the sub-matrix of $A$ formed by the thick columns of $A$ with indices indicated by $J\subset [0, b)$. Similarly, let $A^{[j]}$ and $A^{[J]}$ denote the $j$-th thick row of $A$ and the sub-matrix formed by the thick rows of $A$ with indices in the set $J$, respectively, where $j\in [0, a)$ and $J\subset [0, a)$. For $\ell$ matrices $B_0, B_1, \ldots, B_{\ell-1}$, $\mbox{blkdiag}(B_0, B_1, \ldots, B_{\ell-1})$ denotes the block diagonal matrix 
\begin{equation*}
\begin{pmatrix}
B_0&&&\\& B_1&&\\&&\ddots&\\&&& B_{\ell-1}
\end{pmatrix}.
\end{equation*}

\subsection{MDS Array Codes}
An $[n, n-r, \ell]$ linear array code $\mathcal{C}$ over $\mathbf{F}_q$ has $n-r$ information nodes and $r$ parities nodes in each codeword, with each node (or codeword symbol) being a column vector of length  $\ell$ over $\mathbf{F}_q$, where $\ell$ is referred to as the \textit{sub-packetization level}. The following lemma can be utilized to verify whether an array code is MDS.

\begin{Lemma}(\cite{ye2017explicit})\label{Le-MDSA}
For an $[n, n-r, \ell]$ array code over $\mathbf{F}_q$ admitting the following parity-check matrix
\begin{equation*}
H=\left(\begin{array}{cccc}
    A_{0,0} &  A_{0,1} & \cdots & A_{0,n-1}\\
    A_{1,0} &  A_{1,1} & \cdots & A_{1,n-1}\\
    \vdots &  \vdots & \ddots & \vdots\\
    A_{r-1,0} &  A_{r-1,1} & \cdots & A_{r-1,n-1}
\end{array}\right),
\end{equation*}
where $A_{i,j}$ is an $\ell\times \ell$ matrix over $\mathbf{F}_q$,
it is  an MDS array code if and only if
any $r\times r$ sub-block matrix $H[J]$ of $H$ is nonsingular, where $J\subset [0, n)$ and $|J|=r$.
\end{Lemma}

In \cite{dimakis2010network}, the repair bandwidth $\gamma(d)$ of $[n, n-r, \ell]$ MDS array codes has shown to be
\begin{equation*}
  \gamma(d)\ge\frac{d\ell}{d-(n-r)+1},
\end{equation*}
where $d\in [n-r, n)$ is the number of contacted helper nodes. MDS array codes with the repair bandwidth attaining the above lower bound are said to have the \textit{optimal repair bandwidth}, and are exactly MSR codes. 
If the repair bandwidth of an MDS array code is $(1+\epsilon)$ times the above lower bound where $\epsilon<1$ is a small constant, we say that the MDS array code has \textit{$(1+\epsilon)$-optimal repair bandwidth}, which also referred to as \textit{near-optimal repair bandwidth} in \cite{rawat2018mds}.

Subsequently, the rebuilding access $\Gamma(d)$ of $[n, n-r, \ell]$ MDS array codes has shown to be
\begin{equation*}
  \Gamma(d)\ge\frac{d\ell}{d-(n-r)+1}.
\end{equation*}
MDS array codes with the rebuilding access attaining the above lower bound are said to have the \textit{optimal rebuilding access}. Clearly, MDS array codes with optimal rebuilding access will also have optimal repair bandwidth, but not vice versa.
In the literature, most of the MDS array code constructions focus on $d=n-1$ so as to maximally reduce the repair bandwidth, this is also the setting of this work unless otherwise stated.

\subsection{PMDS Array Codes}
Now we give the formal definition of PMDS array codes.
\begin{Definition}(\cite{holzbaur2021partial})\label{Def-PMDSA}
Let $\mathcal{C}$ be a $[\mu n, \mu(n-r)-s, \ell]$ linear code over $\mathbf{F}_q$, it is said to be a $(\mu, n; r, s)$ PMDS array code with sub-packetization level $\ell$ if
\begin{itemize}
    \item [i)] For any $i\in [0, \mu)$, restricting $\mathcal{C}$ to coordinates in $W_i=[ni, ni+n)$ yields an $[n, n-r, \ell]$ MDS array code over $\mathbf{F}_q$, where $\mathcal{C}|_{W_i}$ is usually referred to as the $i$-th local code and $W_i$ the $i$-th local group.
\item [ii)] For any $E_i\subset W_i$ with $|E_i|=r$,  removing the coordinates of $\mathcal{C}$ in $\bigcup_{i=0}^{\mu-1}E_i$ yields a $[\mu (n-r), \mu (n-r)-s, \ell]$ MDS array code over $\mathbf{F}_q$, i.e., the code $\mathcal{C}$ can correct up to $r$ erasures in $W_i$ for $i\in [0, \mu)$ plus $s$ erasures anywhere.
\end{itemize}
\end{Definition}
Particularly, PMDS array codes over $\mathbf{F}_q$ will be referred to as \emph{PMDS
codes} if $\ell=1$.

From Definition \ref{Def-PMDSA}, it is immediate that every $(\mu, n; r, s)$ PMDS array code with sub-packetization $\ell$ permits a form of the following parity-check matrix
\begin{equation}\label{Eqn-PC-PMDS}
 \mathcal{H}=\left(
                \begin{array}{cccc}
                  H_{0} &  &  &  \\
                   &  H_{1}&  &  \\
                   &  & \ddots &  \\
                   &  &  & H_{\mu-1} \\
                  P_0 & P_1 & \cdots  & P_{\mu-1} \\
                \end{array}
              \right),
\end{equation}
where   $H_i$ is an $r\ell\times n\ell$ matrix and $P_i$ is an $s\ell\times n\ell$ matrix, $\mu, n ,r$ and $s$ denote the number of local groups, the size of each local group (or the code length of each local code), the number of local parities in each local code, and the number of global parities, respectively.

In the following, we introduce several lemmas, which are very useful when checking i) and ii) of Definition \ref{Def-PMDSA}. The immediately following lemma helps to check i) of Definition \ref{Def-PMDSA}, i.e., each local code is MDS.

\begin{Lemma} (Block Vandermonde matrix, \cite[Lemma 14]{ye2017explicit})\label{Lemma-BlockV}
Let $B_0,\ldots, B_{r-1}$ be $\ell\times \ell$ matrices such that
$B_iB_j=B_jB_i$ and $B_i-B_j$ is nonsingular for all $i,j\in [0,r)$ with $i\ne j$, then the matrix
\begin{equation*}
    \begin{pmatrix}
         I &I &\cdots& I\\
         B_0 & B_1 &\cdots & B_{r-1}\\
         \vdots & \vdots &\ddots & \vdots\\
    B_{0}^{r-1} & B_1^{r-1} &\cdots & B_{r-1}^{r-1}\\
    \end{pmatrix}
\end{equation*}
is nonsingular.
\end{Lemma}

The following two lemmas help to verify the non-singularity of some key matrices when checking whether the array codes proposed in Section \ref{sec:codes_r=2_C1} and Section \ref{sec:code_r=3} satisfy ii) of Definition \ref{Def-PMDSA}, respectively.
\begin{Lemma}(\cite[Lemma 2]{gopi2020maximally})\label{Lemma-Det}
Let $C_0, \ldots, C_{s-1}$ be $r\times (r+1)$ matrices and $D_0, \ldots, D_{s-1}$ be $s\times (r+1)$ matrices,  and let $D_i^{(j)}$ be the $j$-th row of $D_i$. Then
\begin{align*}
    &(-1)^{\frac{rs(s-1)}{2}}\det\left(\begin{array}{cccc}
        C_0 & \mathbf{0} & \cdots & \mathbf{0}  \\
       \mathbf{0} &  C_1 & \cdots & \mathbf{0}  \\
       \vdots &  \vdots & \ddots & \vdots  \\
       \mathbf{0} &  \mathbf{0} & \cdots & C_{s-1}  \\
        D_0 &  D_1 & \cdots & D_{s-1}  \\
    \end{array}\right)\\
    =&\det\left(\begin{array}{ccc}
         \det\left(\begin{array}{c}
              C_0  \\
              D_0^{(0)}
         \end{array}\right)& \cdots & \det\left(\begin{array}{c}
              C_{s-1}  \\
              D_{s-1}^{(0)}
         \end{array}\right)\\
         \vdots& \ddots& \vdots\\
        \det\left(\begin{array}{c}
              C_0  \\
              D_0^{(s-1)}
         \end{array}\right)& \cdots & \det\left(\begin{array}{c}
              C_{s-1}  \\
              D_{s-1}^{(s-1)}
         \end{array}\right)
    \end{array}\right).
\end{align*}
\end{Lemma}

\begin{Lemma}(\cite[Lemma 5]{gopi2020maximally})\label{Lemma-Det2}
Let $C_0$ be an $r\times (r+1)$ matrix, $C_1$ be an $r\times (r+2)$ matrix, $D_0$ be a $3\times (r+1)$ matrix and $D_1$ be a $3\times (r+2)$ matrix, and let $D_i^{(j)}$ be the $j$-th row of $D_i$. Then
\begin{align*}
&\det\begin{pmatrix}
     C_0&0\\0&C_1\\D_0&D_1
\end{pmatrix}=0\Longleftrightarrow
\det\begin{pmatrix}
C_0\\D_0^{(0)}
\end{pmatrix}\det\begin{pmatrix}
C_1\\D_1^{(1)}\\D_1^{(2)}
\end{pmatrix}\\
& -\det\begin{pmatrix}
C_0\\D_0^{(1)}
\end{pmatrix}\det\begin{pmatrix}
C_1\\D_1^{(0)}\\D_1^{(2)}
\end{pmatrix}\hspace{-.5mm} +\hspace{-.5mm}\det\begin{pmatrix}
C_0\\D_0^{(2)}
\end{pmatrix}\det\begin{pmatrix}
C_1\\D_1^{(0)}\\D_1^{(1)}
\end{pmatrix}\hspace{-1mm}=0.
\end{align*}
\end{Lemma}

The following lemma helps to calculate some key determinants involved in the RHS of the formulas in Lemmas \ref{Lemma-Det} and \ref{Lemma-Det2} when checking whether the array code proposed in Section \ref{sec:code_r=3} satisfies ii) of Definition \ref{Def-PMDSA}.
\begin{Lemma} (Cauchy--Vandermonde matrix, \cite[Proposition 4.1]{martinez1998fast})\label{Le_det_CV}
Let
\begin{equation*}
V\hspace{-.5mm}=\hspace{-.5mm}\begin{pmatrix}
\frac{1}{c_0-d_0} & \cdots & \frac{1}{c_0-d_{l-1}} &1 & c_0 & \cdots & c_0^{n-l-1}\\
\frac{1}{c_1-d_0} & \cdots & \frac{1}{c_1-d_{l-1}} &1 & c_1 & \cdots & c_1^{n-l-1}\\
\vdots & \ddots & \vdots & \vdots & \vdots & \ddots & \vdots\\
\frac{1}{c_{n-1}-d_0} & \cdots & \frac{1}{c_{n-1}-d_{l-1}} &1 & c_{n-1} & \cdots & c_{n-1}^{n-l-1}\\
\end{pmatrix}
\end{equation*}
be a Cauchy--Vandermonde matrix, then
\begin{equation*}
\det(V)=\frac{\left(\prod_{0\le i<j<n}(c_j-c_i)\right)\left(\prod_{0\le i<j<l}(d_i-d_j)\right)}{\prod_{0\le i< n, 0\le j< l}(c_i-d_j)}.
\end{equation*}
\end{Lemma}

\subsection{Partition of Basis $\{e_0,\cdots,e_{N-1}\}$}\label{subsection:partition}

In this subsection, we revisit a series of particular partitions of a basis set that was proposed in \cite{li2015framework} and \cite{li2021systematic}, which will facilitate the understanding of the new constructions in this paper.

For any two integers  $r,m\ge2$, let $e_0,\cdots,e_{r^m-1}$ be a basis of $\mathbf{F}_q^{r^m}$. For example,
they can be simply set as the standard basis, i.e.,
\begin{equation*}
    e_i=(0,\cdots,0,1,0,\cdots,0),\,\,i\in [0, r^m),
\end{equation*}
where only the $i$-th entry is $1$.

For consistency, we follow the notation in  \cite{li2015framework,li2021systematic}.
Given an integer $0\le a<r^m$, denote by $(a_0,\cdots,a_{m-1})$
its $r$-ary expansion, i.e., $a=\sum\limits_{j=0}^{m-1}r^{m-1-j}a_{j}$.
For $0\le i< m$ and $0\le t<r$, define a subset of $\{e_0,\cdots,e_{r^m-1}\}$ as
\begin{equation}\label{Eqn_Vt}
V_{i,t}=\{e_a|a_i=t, 0\le a< r^m\},
\end{equation}
where $a_i$ is the $i$-th element in the $r$-ary expansion of $a$.

Straightforwardly, $|V_{i,t}|=r^{m-1}$, and  $\{V_{i,0},V_{i,1},\cdots,V_{i,r-1}\}$ is a partition of the set $\{e_0,\cdots,e_{r^m-1}\}$ for any $i\in [0,m)$.
Table \ref{example partition} gives two examples of the set partitions  defined in \eqref{Eqn_Vt}.
{\small
\begin{table}[htbp]
\begin{center}
\caption{(a) and (b) denote the $m$  partitions of  the set $\{e_0,\cdots,e_{r^m-1}\}$    defined by \eqref{Eqn_Vt} for $m=3,r=2$, and $m=2,r=3$, respectively.}
\label{example partition}\begin{tabular}{|c|c|c|c|c|c|c|c|}
\hline $i$ & 0 & 1 & 2 & $i$ & 0 & 1 & 2\\
\hline \multirow{4}{*}{$V_{i,0}$ }&$e_0$&$e_0$&$e_0$&\multirow{4}{*}{$V_{i,1}$ }&$e_4$&$e_2$&$e_1$\\
  & $e_1$&$e_1$&$e_2$ && $e_5$&$e_3$&$e_3$\\
  &$e_2$&$e_4$&$e_4$&& $e_6$&$e_6$&$e_5$\\
  &$e_3$&$e_5$&$e_6$&& $e_7$&$e_7$&$e_7$\\
\hline\multicolumn{8}{c}{\hspace{2mm}(A)}
\end{tabular}\\\vspace{5mm}\setlength{\tabcolsep}{4pt}
\begin{tabular}{|c|c|c|c|c|c|c|c|c|c|}
\hline $i$ & 0 & 1 &  $i$ & 0 & 1 & $i$ & 0 & 1\\
\hline \multirow{3}{*}{$V_{i,0}$ }&$e_0$&$e_0$&\multirow{3}{*}{$V_{i,1}$ }&$e_3$&$e_1$&\multirow{3}{*}{$V_{i,2}$ }&$e_6$&$e_2$\\
  & $e_1$&$e_3$&& $e_4$&$e_4$&& $e_7$&$e_5$\\
  & $e_2$&$e_6$&& $e_5$&$e_7$&& $e_8$&$e_8$\\
\hline\multicolumn{8}{c}{\hspace{6mm}(B)}
\end{tabular}
\end{center}
\end{table}
}

For the convenience of notation,  we also denote by $V_{i,t}$
the $r^{m-1}\times r^m$ matrix
 whose rows are formed by vectors $e_i$
in their corresponding sets, such that $i$  is sorted in ascending order.  For example, when $r=2$ and $m=3$, $V_{1,0}$ can be viewed as a $4\times 8$ matrix as follows
\begin{equation*}
V_{1,0}=\left(e_0^{\top}~ e_1^{\top} ~e_4^{\top}~ e_5^{\top}\right)^{\top},
\end{equation*}
where $\top$ represents the transpose operator.

\subsection{Review of an $[n, n-r,\ell]$ MDS Array Code in \cite{li2021systematic} With Small Sub-packetization Level}

\begin{Construction} (The code $\mathcal{C}_5$ in \cite{li2021systematic})\label{Con-C5}
Let $r, n', n$ be three positive integers, where  $r\ge2$ and $n\ge n'$. For $i\in [0, n)$, denote by $\overline{i}$ the integer in $[0, n')$ such that $i\equiv  \overline{i}\bmod n'$, i.e., $\overline{i}=i\% n'$ for short, with $\%$ denoting the modulo operation. Let $\ell=r^{n'}$, then
an $[n,n-r,\ell]$ array code over $\mathbf{F}_q$ is defined by the following parity-check matrix
\begin{equation}\label{Eqn-PC-C5}
H=\begin{pmatrix}
I & I & \cdots & I\\
A_0 & A_1 & \cdots & A_{n-1}\\
\vdots & \vdots & \ddots & \vdots\\
A_0^{r-1} & A_1^{r-1} & \cdots & A_{n-1}^{r-1}\\
\end{pmatrix},
\end{equation}
where $A_i$, $i\in [0,n)$ satisfy
\begin{equation}\label{Eqn Hadamard coding matrix}
\left(
                     \begin{array}{c}
                       V_{\overline{i},0} \\
                       V_{\overline{i},1} \\
                       \vdots \\
                       V_{\overline{i},r-1}
                     \end{array}
                   \right)A_i=\left(
                     \begin{array}{c}
                      \lambda_{i,0} V_{\overline{i},0} \\
                      \lambda_{i,1} V_{\overline{i},1} \\
                       \vdots \\
                      \lambda_{i,r-1} V_{\overline{i},r-1}
                     \end{array}
                   \right),
\end{equation}
with $\lambda_{i,t} \in \mathbf{F}_q\backslash\{0\}$ and $V_{\overline{i}, t}$ being defined by \eqref{Eqn_Vt} for $t\in [0, r)$.
\end{Construction}

Clearly from \eqref{Eqn_Vt} and \eqref{Eqn Hadamard coding matrix}, we have that $A_i$ is a diagonal matrix and the $a$-th row of $A_i$ is
\begin{equation}\label{Eqn-row-of-A}
e_aA_i=\lambda_{i, a_{\overline{i}}}e_a,
\end{equation}
where $i\in [0, n)$, $a\in [0, \ell)$, and $a_{\overline{i}}$ denotes the $\overline{i}$-th element in the $r$-ary expansion of $a$.

In the following, we revisit the results related to the code $\mathcal{C}_5$ in \cite{li2021systematic}.

\begin{Lemma} (Theorems 13 and 14 of \cite{li2021systematic})\label{lem repair}
The array code in Construction \ref{Con-C5} is MDS with the repair bandwidth $\gamma_i$ of node $i$ ($i\in [0, n)$) being
\begin{equation}\label{Eqn-RB-C5}
  \gamma_i =\left\{\hspace{-2.5mm}
                      \begin{array}{ll}
                      (1+\frac{(\lceil\frac{n}{n'}\rceil-1)(r-1)}{n-1})\frac{\ell}{r}(n-1), &\mbox{if~} 0\le i\% n'<n\%n', \\
                       (1+\frac{(\lfloor\frac{n}{n'}\rfloor-1)(r-1)}{n-1})\frac{\ell}{r}(n-1), & \mbox{otherwise},
                      \end{array}
                    \right.
\end{equation} if the following requirements can be satisfied
\begin{itemize}
  \item [R1.] $\lambda_{i,u}\ne \lambda_{j,v}$ for all $u,v\in[0,r)$ and $i,j\in[0,n)$ with $j\not\equiv i \bmod n'$,
  \item [R2.] $\lambda_{i,u}\ne \lambda_{i+gn',u}$ for all $u\in[0,r)$, $g\in[1,\lceil \frac{n}{n'}\rceil)$, $i\in[0,n')$   with $i+gn'<n$,
  \item [R3.] $\lambda_{i,0},\lambda_{i,1},\cdots,\lambda_{i,r-1}$ are pairwise distinct for every $i\in [0,n)$.
\end{itemize}
The rebuilding access $\Gamma_i$ of node $i$ is
\begin{equation*}
\Gamma_i=\ell(n-1)~\mbox{for}~i\in [0, n)
\end{equation*}
when the repair bandwidth achieves the value in \eqref{Eqn-RB-C5}.
\end{Lemma}

\begin{Lemma}(Theorem 15 of \cite{li2021systematic})\label{lem-C5-q}
The three requirements R1-R3 in Lemma \ref{lem repair} can be satisfied if the finite field $\mathbf{F}_q$ contains at least $\Phi$  nonzero elements, where
\begin{equation}\label{Eqn-Phi}
  \Phi=\left\{
                      \begin{array}{ll}
                       rn'(\lceil\frac{n}{rn'}\rceil-1)+(n\% n')r, & \mbox{\ if\ }0<n\% (rn')<n',\\
                       rn'\lceil\frac{n}{rn'}\rceil, &\mbox{otherwise}.
                      \end{array}
                    \right.
\end{equation}
Furthermore,
$A_iA_j=A_jA_i$ and $A_i-A_j$ is nonsingular for all $i,j\in [0,n)$ with $i\ne j$.
\end{Lemma}

\section{New PMDS array code constructions with Two Global Parities}\label{sec:codes_r=2}

In this section, we present two PMDS array code constructions with two global parities. The first one allows an arbitrary number $r$ of local parities, and has $(1+\epsilon)$-optimal repair bandwidth but high rebuilding access, while the second one has both $(1+\epsilon)$-optimal repair bandwidth and smaller rebuilding access when compared with the first one, but can only support two local parities, i.e, $r=2$.

\subsection{A New $(\mu, n; r, s=2)$ PMDS Array Code Construction}\label{sec:codes_r=2_C1}

\begin{Construction}\label{Con-new1}
Let $\mu, r, n', n$ be four positive integers, where  $\mu, r\ge2$ and $n>n'$, and let $\ell=r^{n'}$. We construct a new $[\mu n, \mu (n-r)-2,\ell]$ array code over $\mathbf{F}_q$ with the parity-check matrix having the form as in \eqref{Eqn-PC-PMDS}, where
\begin{equation}\label{Eqn-PC-new2}
H_i=H,
\end{equation}
and
\begin{equation}\label{Eqn-PC-new3}
P_i=\left(\begin{array}{cccc}
    A_0^{r} & A_1^{r} & \cdots & A_{n-1}^{r} \\
    \theta_{i}A_0^{-1} & \theta_{i}A_1^{-1} & \cdots & \theta_{i}A_{n-1}^{-1}
\end{array}\right)
\end{equation}
for $i\in [0, \mu)$,
where $H$ is defined in \eqref{Eqn-PC-C5}, i.e., the parity-check matrix of the code in Construction \ref{Con-C5},
$A_i$ is an $\ell\times \ell$ matrix defined in \eqref{Eqn Hadamard coding matrix}, $\theta_i\in \mathbf{F}_q\backslash\{0\}$ for $i\in [0, \mu)$.
\end{Construction}

\begin{Remark}
Note that in both Construction \ref{Con-new1} above and Construction A in \cite{holzbaur2021partial}, a scalar multiplier is employed to distinguish different global parity-check blocks $P_0, P_1,\ldots, P_{\mu-1}$ in \eqref{Eqn-PC-PMDS}. 

Precisely, in \eqref{Eqn-PC-new3}, we employ independent scalar multipliers $\theta_0,\theta_1,\ldots,\theta_{\mu-1}$ in $P_0, P_1,\ldots, P_{\mu-1}$, which is motivated by the constructions in \cite{gopi2020maximally}. In addition, the sub-block matrix formed by any $r+2$ thick columns of $\begin{pmatrix}
H_i\\
P_i
\end{pmatrix}
$ or any $r+1$ thick columns of $\begin{pmatrix}
H_i\\
P_i^{[t]}
\end{pmatrix}
$ ($t=0,1$) is equivalent to a block Vandermonde matrix after row permutation and scaling. Thus the determinants of the sub-block matrices mentioned above can be easily calculated, which greatly facilitates the proof that Construction \ref{Con-new1} gives a PMDS code. 

Whereas in Construction A in \cite{holzbaur2021partial}, the multipliers are $\beta^{-N}, \beta^{-2N}, \ldots, \beta^{-\mu N}$, where $N$ is an integer  and $\beta$ has order at least $\mu N$ in $\mathbf{F}_q$. It was proved in \cite{holzbaur2021partial} that $N$ should be larger than a threshold to guarantee the code is PMDS, which leads to a larger finite field than ours.
\end{Remark}

\begin{figure*}
\begin{equation}\label{Eqn_Ba_C1}
B_a=\begin{pmatrix}
      1 & 1 & \cdots & 1  &&&& \\
    \lambda_{j_0,a_{\overline{j_0}}} &  \lambda_{j_1,a_{\overline{j_1}}} & \cdots & \lambda_{j_r,a_{\overline{j_r}}} &&&&\\
    \vdots & \vdots & \ddots & \vdots&&&&\\
     \lambda_{j_0,a_{\overline{j_0}}}^{r-1} &  \lambda_{j_1,a_{\overline{j_1}}}^{r-1} & \cdots & \lambda_{j_r,a_{\overline{j_r}}}^{r-1} &&&&\\
  &&&&  1 & 1 & \cdots & 1   \\
  &&&&  \lambda_{t_0,a_{\overline{t_0}}} &  \lambda_{t_1,a_{\overline{t_1}}} & \cdots & \lambda_{t_r,a_{\overline{t_r}}}\\
  &&&&  \vdots & \vdots & \ddots & \vdots\\
   &&&& \lambda_{t_0,a_{\overline{t_0}}}^{r-1} &  \lambda_{t_1,a_{\overline{t_1}}}^{r-1} & \cdots & \lambda_{t_r,a_{\overline{t_r}}}^{r-1}\\
   \lambda_{j_0,a_{\overline{j_0}}}^{r} &  \lambda_{j_1,a_{\overline{j_1}}}^{r} & \cdots & \lambda_{j_r,a_{\overline{j_r}}}^{r} &  \lambda_{t_0,a_{\overline{t_0}}}^{r} &  \lambda_{t_1,a_{\overline{t_1}}}^{r} & \cdots & \lambda_{t_r,a_{\overline{t_r}}}^{r} \\
    \theta_i    \lambda_{j_0,a_{\overline{j_0}}}^{-1} &  \theta_i \lambda_{j_1,a_{\overline{j_1}}}^{-1} & \cdots & \theta_i \lambda_{j_r,a_{\overline{j_r}}}^{-1} &  \theta_{k}\lambda_{t_0,a_{\overline{t_0}}}^{-1} &  \theta_{k}\lambda_{t_1,a_{\overline{t_1}}}^{-1} & \cdots & \theta_{k}\lambda_{t_r,a_{\overline{t_r}}}^{-1}
\end{pmatrix}, ~a\in [0, \ell).
\end{equation}
\hrule
\end{figure*}

\begin{Theorem}\label{Eqn_Thm_C1}
Let $q>\mu\Phi$ be a prime power such that there exists a multiplicative subgroup $G$ of $\mathbf{F}_q\backslash\{0\}$ of size at least $\Phi$ and with at least $\mu$ cosets, where $\Phi$ is defined in \eqref{Eqn-Phi}. Choosing $\lambda_{i,t}$, $i\in [0, n)$, $t\in [0, r)$ from $G$, then the code in Construction \ref{Con-new1} is a PMDS array code if $\theta_0,\ldots,\theta_{\mu-1}$ are elements from distinct cosets of $G$, where the repair bandwidth and the rebuilding access of node $in+j$ are
\begin{equation*}
 \gamma_{in+j} \hspace{-1mm}=\hspace{-1.4mm}\left\{\hspace{-2.4mm}
                   \begin{array}{ll}
(1\hspace{-.4mm}+\hspace{-.4mm}\frac{(\lceil\frac{n}{n'}\rceil-1)(r-1)}{n-1})\frac{\ell}{r}(n-1), \hspace{-.9mm}&\mbox{if\ } 0\hspace{-.4mm}\le\hspace{-.4mm} j\% n' \hspace{-.6mm}< \hspace{-.6mm}n\%n', \\
                       (1\hspace{-.4mm}+\hspace{-.4mm}\frac{(\lfloor\frac{n}{n'}\rfloor-1)(r-1)}{n-1})\frac{\ell}{r}(n-1),\hspace{-1mm} & \mbox{otherwise},
                      \end{array}
                    \right.
\end{equation*}
and
\begin{equation*}
\Gamma_{in+j}=\ell(n-1)
\end{equation*}
for $i\in [0, \mu)$ and $j\in [0, n)$, respectively.
\end{Theorem}

\begin{proof}
By \eqref{Eqn-PC-PMDS} and \eqref{Eqn-PC-new2}, we see that each local code is an MDS array code defined by the parity-check matrix $H$ of the code in Construction \ref{Con-C5}.
Thus, the statement on the repair bandwidth and rebuilding access follows from Lemmas \ref{lem repair} and \ref{lem-C5-q}, and i) of Definition \ref{Def-PMDSA} is satisfied.

Now let us check ii) of Definition \ref{Def-PMDSA}.
Suppose that there are $r$ failed nodes in every local group and two more anywhere. We can easily repair the nodes in the local groups with at most $r$ node failures since each local code is an MDS array code. Now we are left with two cases:
\begin{itemize}
    \item Both the two extra failed nodes are in the same local group, say group $i$ and assume that nodes $in+j_0,\ldots,in+j_{r+1}$ are failed, where $i\in [0, \mu)$ and $0\le j_0<\cdots<j_{r+1}< n$. Let $J=\{j_0,\ldots,j_{r+1}\}$, then the original file can be reconstructed if the following matrix
\begin{equation*}
\left(\begin{array}{c}
     H_i[J]  \\
    P_i[J]
\end{array}\right)=\left(
\begin{array}{cccc}
      I & I & \cdots & I  \\
    A_{j_0} & A_{j_1} & \cdots & A_{j_{r+1}}\\
    \vdots & \vdots & \ddots & \vdots\\
    A_{j_0}^{r-1} & A_{j_1}^{r-1} & \cdots & A_{j_{r+1}}^{r-1} \\
    A_{j_0}^{r} & A_{j_1}^{r} & \cdots & A_{j_{r+1}}^{r} \\
    \theta_i A_{j_0}^{-1} &\theta_i A_{j_1}^{-1} & \cdots &\theta_i A_{j_{r+1}}^{-1}
\end{array}
\right)
\end{equation*}
is of full rank.

Note that
\begin{align*}
&\begin{psmallmatrix}
 &  &  & \theta_i^{-1} I  \\
      I &  &  &   \\
        & \ddots &  & \\
     & & I &
\end{psmallmatrix}\begin{psmallmatrix}
     H_i[J]  \\
    P_i[J]
\end{psmallmatrix} \begin{psmallmatrix}
      A_{j_0} &  &  &   \\
        & A_{j_1} &  & \\
     & & \ddots & \\
 &  &  & A_{j_{r+1}}
\end{psmallmatrix}\\
=&\left(
\begin{array}{cccc}
      I & I & \cdots & I  \\
    A_{j_0} & A_{j_1} & \cdots & A_{j_{r+1}}\\
    \vdots & \vdots & \ddots & \vdots\\
    A_{j_0}^{r+1} & A_{j_1}^{r+1} & \cdots & A_{j_{r+1}}^{r+1}
\end{array}
\right),
\end{align*}
which is a block Vandermonde matrix and thus is nonsingular by Lemmas \ref{Lemma-BlockV} and \ref{lem-C5-q}. Therefore, $\left(\begin{array}{c}
     H_i[J]  \\
    P_i[J]
\end{array}\right)$ is nonsingular and the original file can be reconstructed.

\item
The two extra node failures are in two different groups, say group $i$ and $k$, and assume that nodes $in+j_0,\ldots,in+j_{r}$ and nodes $kn+t_0,\ldots,kn+t_{r}$ are failed, where $0\le i<k<\mu$, $0\le j_0<\cdots<j_{r}< n$, and $0\le t_0<\cdots<t_{r}< n$. Let $J=\{j_0,\ldots,j_{r}\}$ and $T=\{t_0,\ldots,t_{r}\}$, then the original file can be reconstructed if the following matrix
\begin{align*}
\hat{\mathcal{H}}
=&\begin{pmatrix}
H_i[J]&\\
& H_{k}[T]\\
P_i[J]& P_{k}[T]
\end{pmatrix}\\
=&
\begin{psmallmatrix}
      I & I & \cdots & I  &&&& \\
    A_{j_0} & A_{j_1} & \cdots & A_{j_{r}}&&&&\\
    \vdots & \vdots & \ddots & \vdots&&&&\\
    A_{j_0}^{r-1} & A_{j_1}^{r-1} & \cdots & A_{j_{r}}^{r-1} &&&&\\
  &&&&  I & I & \cdots & I   \\
  &&&&  A_{t_0} & A_{t_1} & \cdots & A_{t_{r}}\\
  &&&&  \vdots & \vdots & \ddots & \vdots\\
   &&&& A_{t_0}^{r-1} & A_{t_1}^{r-1} & \cdots & A_{t_{r}}^{r-1} \\
   A_{j_0}^{r} & A_{j_1}^{r} & \cdots & A_{j_{r}}^{r} &  A_{t_0}^{r} & A_{t_1}^{r} & \cdots & A_{t_{r}}^{r} \\
    \theta_i A_{j_0}^{-1} &\theta_i A_{j_1}^{-1} & \cdots &\theta_i A_{j_{r}}^{-1} & \theta_{k} A_{t_0}^{-1} &\theta_{k} A_{t_1}^{-1} & \cdots &\theta_{k} A_{t_{r}}^{-1}
\end{psmallmatrix}
\end{align*}
is nonsingular.

Since each block entry in the parity-check matrix $\hat{\mathcal{H}}$ is an $\ell\times \ell$ diagonal matrix, by swapping the rows and columns of $\hat{\mathcal{H}}$, $\hat{\mathcal{H}}$ is equivalent to
\begin{equation*}
\mbox{blkdiag}(B_0, B_1, \ldots, B_{\ell-1})
\end{equation*}
under elementary transformation, which has the same rank as $\hat{\mathcal{H}}$, where $B_a$ is formed by the $a, a+\ell, \ldots, a+(2r+1)\ell$-th rows and the $a, a+\ell, \ldots, a+(2r+1)\ell$-th columns of $\hat{\mathcal{H}}$, and can be expressed as in \eqref{Eqn_Ba_C1} according to \eqref{Eqn-row-of-A}.

For $a\in [0, \ell)$, by Lemma \ref{Lemma-Det}, we have \eqref{Eqn_detBa_C1}.
\begin{figure*}
\begin{equation}\label{Eqn_detBa_C1}
(-1)^r\det(B_a)
=\det\left(\hspace{-2mm}\begin{array}{cc}
         \det\begin{pmatrix}
             1 & 1 & \cdots & 1   \\
    \lambda_{j_0,a_{\overline{j_0}}} &  \lambda_{j_1,a_{\overline{j_1}}} & \cdots & \lambda_{j_r,a_{\overline{j_r}}} \\
    \vdots & \vdots & \ddots & \vdots\\
     \lambda_{j_0,a_{\overline{j_0}}}^{r} &  \lambda_{j_1,a_{\overline{j_1}}}^{r} & \cdots & \lambda_{j_r,a_{\overline{j_r}}}^{r} \\
         \end{pmatrix}&  \det\begin{pmatrix}
              1 & 1 & \cdots & 1   \\
    \lambda_{t_0,a_{\overline{t_0}}} &  \lambda_{t_1,a_{\overline{t_1}}} & \cdots & \lambda_{t_r,a_{\overline{t_r}}}\\
    \vdots & \vdots & \ddots & \vdots\\
    \lambda_{t_0,a_{\overline{t_0}}}^{r} &  \lambda_{t_1,a_{\overline{t_1}}}^{r} & \cdots & \lambda_{t_r,a_{\overline{t_r}}}^{r}\\
        \end{pmatrix}\\
        \det\begin{pmatrix}
             1 & 1 & \cdots & 1   \\
    \lambda_{j_0,a_{\overline{j_0}}} &  \lambda_{j_1,a_{\overline{j_1}}} & \cdots & \lambda_{j_r,a_{\overline{j_r}}} \\
    \vdots & \vdots & \ddots & \vdots\\
     \lambda_{j_0,a_{\overline{j_0}}}^{r-1} &  \lambda_{j_1,a_{\overline{j_1}}}^{r-1} & \cdots & \lambda_{j_r,a_{\overline{j_r}}}^{r-1} \\
    \theta_i \lambda_{j_0,a_{\overline{j_0}}}^{-1} &  \theta_i\lambda_{j_1,a_{\overline{j_1}}}^{-1} & \cdots & \theta_i \lambda_{j_r,a_{\overline{j_r}}}^{-1}
         \end{pmatrix} &  \det\begin{pmatrix}
              1 & 1 & \cdots & 1   \\
    \lambda_{t_0,a_{\overline{t_0}}} &  \lambda_{t_1,a_{\overline{t_1}}} & \cdots & \lambda_{t_r,a_{\overline{t_r}}}\\
    \vdots & \vdots & \ddots & \vdots\\
    \lambda_{t_0,a_{\overline{t_0}}}^{r-1} &  \lambda_{t_1,a_{\overline{t_1}}}^{r-1} & \cdots & \lambda_{t_r,a_{\overline{t_r}}}^{r-1}\\
  \theta_{k}  \lambda_{t_0,a_{\overline{t_0}}}^{-1} &  \theta_{k}\lambda_{t_1,a_{\overline{t_1}}}^{-1} & \cdots & \theta_{k}\lambda_{t_r,a_{\overline{t_r}}}^{-1}
         \end{pmatrix}
    \end{array}\hspace{-2mm}\right).
\end{equation}
\hrule
\end{figure*}
By factoring out the nonzero Vandermonde determinant from each column of the determinant in \eqref{Eqn_detBa_C1}, we have
 \begin{align*}
& (-1)^r\det(B_a)\ne 0 \\
 \iff&\det\begin{pmatrix}
  \lambda_{j_0,a_{\overline{j_0}}}\cdots \lambda_{j_r,a_{\overline{j_r}}}   & \lambda_{t_0,a_{\overline{t_0}}}\cdots \lambda_{t_r,a_{\overline{t_r}}} \\
 \theta_i    & \theta_{k}
 \end{pmatrix}
\ne 0,
 \end{align*}
  and the last inequality holds since both $\lambda_{j_0,a_{\overline{j_0}}}\lambda_{j_1,a_{\overline{j_1}}}\cdots \lambda_{j_r,a_{\overline{j_r}}}$ and $\lambda_{t_0,a_{\overline{t_0}}}\lambda_{t_1,a_{\overline{t_1}}}\cdots \lambda_{t_r,a_{\overline{t_r}}}$ are in the subgroup $G$, while  $\theta_i$ and $\theta_{k}$ are in different cosets of $G$.

Therefore, $\mbox{blkdiag}(B_0, B_1, \ldots, B_{\ell-1})$ and thus $\hat{\mathcal{H}}$ is nonsingular, and the original file can be reconstructed.
\end{itemize}

This finishes the proof.
\end{proof}

\subsection{A New $(\mu, n; r=2, s=2)$ PMDS Array Code Construction With Small Rebuilding Access}

Partly motivated by the construction of MSR codes with optimal rebuilding access in \cite{chen2020explicit}, we present the second PMDS array code construction, which deploys both triangular matrices and diagonal matrices as building blocks. This is the first time to  use non-diagonal matrices as building blocks in PMDS array codes besides diagonal matrices, and thus leads to smaller rebuilding access (when normalized by the file size) and a smaller sub-packetization level, although it only supports two local parities.

\begin{Construction}\label{Con-PMDS2}
For convenience of notation, we assume that $n$ is even. Construct a  $[\mu n, \mu(n-2)-2,\ell=2]$ linear array code
over $\mathbf{F}_q$ with the parity-check matrix having the form as in \eqref{Eqn-PC-PMDS}, where
\begin{equation}\label{Eqn-PC-new4}
H_i=H=\left(\begin{array}{cccc}
    I & I & \cdots & I  \\
    A_0 & A_1 & \cdots & A_{n-1}
\end{array}\right)
\end{equation}
and
\begin{equation}\label{Eqn-PC-new5}
P_i=\left(\begin{array}{cccc}
    \lambda_0^{2}I & \lambda_1^{2}I & \cdots & \lambda_{n-1}^{2}I \\
\theta_i\lambda_0^{-1}I &\theta_i \lambda_1^{-1}I & \cdots & \theta_i\lambda_{n-1}^{-1}I
\end{array}\right)
\end{equation}
for $i\in [0, \mu)$, where $I$ denotes the identity matrix of order $2$,
\begin{eqnarray}\label{Eqn-PC-A}
A_{i}= \left\{
   \begin{array}{cc}
    \left(
   \begin{array}{cc}
   \lambda_i & 1\\
   0&\lambda_i
   \end{array}
 \right), & \mbox{if}~ 2\mid i,\\
 \left(
   \begin{array}{cc}
   \lambda_i & 0\\
   0&\lambda_i
   \end{array}
 \right), & \mbox{otherwise,}
   \end{array}
 \right.
\end{eqnarray}
and $\theta_0,\ldots,\theta_{\mu-1}$ and   $\lambda_0,\ldots,\lambda_{n-1}$ are pairwise distinct nonzero elements in $\mathbf{F}_q$, respectively.
\end{Construction}

\begin{Theorem}\label{Thm-band-local}
Both the repair bandwidth and the rebuilding access of the array code in Construction \ref{Con-PMDS2} are $\frac{3n}{2}-2$.
\end{Theorem}
\begin{proof}
For a given $b$, where $b\in [0, \mu)$, let $f_0,\ldots,f_{n-1}\in \mathbf{F}_q^2$ denote the data stored in the $n$ nodes of the $b$-th local group.
Then the $b$-th local code is subject to the following parity-check equations.
\begin{eqnarray}
\label{Eqn-pc1}f_0+f_1+\cdots+f_{n-1}&=&\mathbf{0},\\
\label{Eqn-pc2}A_0f_0+A_1f_1+\cdots+A_{n-1}f_{n-1}&=&\mathbf{0}.
\end{eqnarray}
Suppose node $i$ of the $b$-th local group is failed, then we can repair it by contacting all the remaining nodes in the same local group. In the following, we analyze the repair bandwidth and the rebuilding access.

\begin{itemize}
\item   If $2\mid i$,
then let $e_0=\left(
   \begin{array}{cc}
   1 & 0\\
   \end{array}
 \right)$. By multiplying $e_0$ with \eqref{Eqn-pc1}  and \eqref{Eqn-pc2} from the left, we get
\begin{eqnarray*}
\left(
   \begin{array}{c}
   e_0\\
   e_0A_i
   \end{array}
 \right)f_i=-\sum\limits_{j=0,j\ne i}^{n-1}\left(
   \begin{array}{c}
   e_0\\
   e_0A_j
   \end{array}
 \right)f_j.
\end{eqnarray*}
Note that
\begin{equation*}
    \mbox{rank}(\left(
   \begin{array}{c}
   e_0\\
   e_0A_i
   \end{array}
 \right))=\mbox{rank}(\left(
   \begin{array}{cc}
   1&0\\
   \lambda_i&1
   \end{array}
 \right))=2,
\end{equation*}
and
\begin{align*}
 &\mbox{rank}(\left(
   \begin{array}{c}
   e_0\\
   e_0A_j
   \end{array}
 \right))\\=&\left\{
   \begin{array}{cc}
   \mbox{rank}(\left(
   \begin{array}{ll}
   1&0\\
   \lambda_j&1
   \end{array}
 \right))=2, & \mbox{if}~ 2\mid j,\\[12pt]
\mbox{rank}(\left(
   \begin{array}{cc}
   1&0\\
   \lambda_j&0
   \end{array}
 \right))=1,& \mbox{otherwise}.
      \end{array}
 \right.
\end{align*}

Thus the repair bandwidth $\gamma_i$ of node $i$
is
\begin{equation*}
 \gamma_i=\sum\limits_{j=0,j\ne i}^{n-1}\mbox{rank}(\left(\hspace{-1mm}
   \begin{array}{c}
   e_0\\
   e_0A_j
   \end{array}
 \hspace{-1mm}\right)) =(\frac{n}{2}-1) \times 2+\frac{n}{2}=\frac{3n}{2}-2,
\end{equation*}
and the rebuilding access $\Gamma_i$ of node $i$
is
\begin{equation*}
 \Gamma_i=\sum\limits_{j=0,j\ne i}^{n-1}N_c(\left(\hspace{-1mm}
   \begin{array}{c}
   e_0\\
   e_0A_j
   \end{array}
\hspace{-1mm} \right)) =(\frac{n}{2}-1) \times 2+\frac{n}{2}=\frac{3n}{2}-2, \end{equation*}
where $N_c(A)$ denotes the number of nonzero columns of the matrix $A$.

\item If $2\nmid i$, then let $S=\left(\begin{array}{cc}
    1 & 0 \\
    0 & 1
\end{array}\right)$ and $S'=\left(\begin{array}{cc}
    0 & 0 \\
    1 & 0
\end{array}\right)$. By $S\times$  \eqref{Eqn-pc1} +  $S'\times$\eqref{Eqn-pc2}, we obtain
\begin{equation*}
   (S+S'A_i)f_i=-\sum\limits_{j=0,j\ne i}^{n-1} (S+S'A_j)f_j.
\end{equation*}
Note that
\begin{align*}
S+S'A_i=&\left(\begin{array}{cc}
    1 & 0 \\
    0 & 1
\end{array}\right)+\left(\begin{array}{cc}
    0 & 0 \\
    1 & 0
\end{array}\right) \left(
   \begin{array}{cc}
   \lambda_i & 0\\
   0&\lambda_i
   \end{array}
 \right)\\=&\left(\begin{array}{cc}
    1 & 0 \\
    \lambda_i & 1
\end{array}\right)
\end{align*}
and
\begin{align*}
 &S+S'A_j\\=& \left\{\hspace{-2mm}
   \begin{array}{cc}
\left(\begin{array}{cc}
    1 & 0 \\
    0 & 1
\end{array}\right)+\left(\begin{array}{cc}
    0 & 0 \\
    1 & 0
\end{array}\right) \hspace{-1mm}\left(
   \begin{array}{cc}
   \lambda_j & 0\\
   0&\lambda_j
   \end{array}
 \right), & \mbox{if~} 2\nmid j,   \\
\left(\begin{array}{cc}
    1 & 0 \\
    0 & 1
\end{array}\right)+\left(\begin{array}{cc}
    0 & 0 \\
    1 & 0
\end{array}\right)\hspace{-1mm} \left(
   \begin{array}{cc}
   \lambda_j &1\\
   0&\lambda_j
   \end{array}
 \right),& \mbox{otherwise},
   \end{array}\right.\\=& \left\{\hspace{-2mm}
   \begin{array}{cc}
\left(\begin{array}{cc}
    1 & 0 \\
    \lambda_j & 1
\end{array}\right), & \mbox{if~} 2\nmid j,   \\
\left(\begin{array}{cc}
    1 & 0 \\
    \lambda_j & 0
\end{array}\right),& \mbox{otherwise}.
   \end{array}\right.
\end{align*}
Thus the repair bandwidth $\gamma_i$ of node $i$
is
\begin{equation*}
 \gamma_i=\sum\limits_{j=0,j\ne i}^{n-1}\mbox{rank}(S+S'A_j)= (\frac{n}{2}-1) \times 2+\frac{n}{2}=\frac{3n}{2}-2,
\end{equation*}
while the rebuilding access $\Gamma_i$ of node $i$
is
\begin{equation*}
 \Gamma_i=\sum\limits_{j=0,j\ne i}^{n-1}N_c(S+S'A_j) =(\frac{n}{2}-1) \times 2+\frac{n}{2}=\frac{3n}{2}-2.
\end{equation*}
\end{itemize}
This completes the proof.
\end{proof}

\begin{Remark}
When $n$ is odd in Construction \ref{Con-PMDS2}, similar to the proof of Theorem \ref{Thm-band-local}, we have that 
both the repair bandwidth and the rebuilding access of node $i$ ($i\in [0, n)$) in local group $j$ ($j\in [0, \mu)$) of the array code in Construction \ref{Con-PMDS2} are
\begin{equation*}
\gamma_i=\Gamma_i=\left\{\begin{array}{cc}
\frac{3n-3}{2}, & {\rm if}~ 2\mid i,\\
\frac{3n-5}{2}, & \rm{otherwise}.\end{array}\right.
\end{equation*}
\end{Remark}

\begin{figure*}
\begin{equation}\label{Eqn_hatH_C2}
 \hat{\mathcal{H}}
=\left(
   \begin{array}{cccccccccccc}
     1 &0 & 1 &0  &  1&0  &0& 0& 0 & 0 & 0 & 0  \\
    0  &1 &0 & 1 &0  &  1&0  &0& 0& 0 & 0 & 0     \\
     \lambda_{j_0} &x_0 & \lambda_{j_1} &x_1  & \lambda_{j_2}& x_2  &0& 0& 0 & 0 & 0 & 0  \\
    0  &\lambda_{j_0} &0 & \lambda_{j_1} &0  &  \lambda_{j_2}&0  &0& 0& 0 & 0 & 0     \\
    0& 0& 0 & 0 & 0 & 0 & 1 &0 & 1 &0  &  1&0    \\
  0  &0& 0& 0 & 0 & 0  &0  &1 &0 & 1 &0  &  1     \\
  0& 0& 0 & 0 & 0 & 0 & \lambda_{t_0} &x_3 & \lambda_{t_1} &x_4  & \lambda_{t_2}& x_5    \\
0  &0& 0& 0 & 0 & 0 &   0  &\lambda_{t_0} &0 & \lambda_{t_1} &0  &  \lambda_{t_2}     \\
 \lambda_{j_0}^2&0 &\lambda_{j_1}^2&0 &\lambda_{j_2}^2&0 &\lambda_{t_0}^2&0 &\lambda_{t_1}^2&0 &\lambda_{t_2}^2&0\\
0& \lambda_{j_0}^2&0 &\lambda_{j_1}^2&0 &\lambda_{j_2}^2&0 &\lambda_{t_0}^2&0 &\lambda_{t_1}^2&0 &\lambda_{t_2}^2\\
 \theta_i \lambda_{j_0}^{-1}&0 &\theta_i\lambda_{j_1}^{-1}&0 &\theta_i\lambda_{j_2}^{-1}&0 &\theta_k\lambda_{t_0}^{-1}&0 &\theta_k\lambda_{t_1}^{-1}&0 &\theta_k\lambda_{t_2}^{-1}&0\\
0&\theta_i \lambda_{j_0}^{-1}&0 &\theta_i\lambda_{j_1}^{-1}&0 &\theta_i\lambda_{j_2}^{-1}&0 &\theta_k\lambda_{j_0}^{-1}&0 &\theta_k\lambda_{j_1}^{-1}&0 &\theta_k\lambda_{j_2}^{-1}\\
   \end{array}
 \right).
\end{equation}
\hrule
\end{figure*}
\begin{figure*}
\begin{equation}\label{Eqn_hatH'_C3}
\left(
   \begin{array}{cccccccccccc}
     1  & 1   &  1  &0& 0& 0 & 0 & 0 & 0& 0 & 0 & 0  \\
     \lambda_{j_0}  & \lambda_{j_1}   & \lambda_{j_2}  &0& 0& 0 &x_0 &x_1 & x_2 & 0 & 0 & 0  \\
    0& 0& 0 & 1 & 1  &  1& 0 & 0 & 0  &0  &0 &0    \\
  0& 0& 0 & \lambda_{t_0} & \lambda_{t_1} & \lambda_{t_2}&  0 & 0 & 0  &x_3  &x_4  & x_5    \\
 \lambda_{j_0}^2 &\lambda_{j_1}^2 &\lambda_{j_2}^2 &\lambda_{t_0}^2 &\lambda_{t_1}^2 &\lambda_{t_2}^2& 0& 0 & 0 & 0 & 0& 0\\
  \theta_i\lambda_{j_0}^{-1} &\theta_i\lambda_{j_1}^{-1} &\theta_i\lambda_{j_2}^{-1} &\theta_k\lambda_{t_0}^{-1} &\theta_k\lambda_{t_1}^{-1} &\theta_k\lambda_{t_2}^{-1}& 0& 0 & 0 & 0 & 0& 0\\
0   &0  &0 &0  &0& 0 &1&1&1& 0 & 0 & 0     \\
 0   &0  &0  & 0  &0& 0 &\lambda_{j_0} & \lambda_{j_1} &  \lambda_{j_2}& 0 & 0 & 0     \\
  0  &0& 0& 0 & 0 & 0  &0  &0 &0&1 & 1   &  1     \\
  0  &0& 0& 0 & 0 & 0 &   0 &0&0 &\lambda_{t_0}  & \lambda_{t_1}   &  \lambda_{t_2}     \\
  0&0&0&0&0&0& \lambda_{j_0}^2 &\lambda_{j_1}^2 &\lambda_{j_2}^2 &\lambda_{t_0}^2 &\lambda_{t_1}^2 &\lambda_{t_2}^2\\
  0&0&0&0&0&0& \theta_i\lambda_{j_0}^{-1} &\theta_i\lambda_{j_1}^{-1} &\theta_i\lambda_{j_2}^{-1} &\theta_k\lambda_{t_0}^{-1} &\theta_k\lambda_{t_1}^{-1} &\theta_k\lambda_{t_2}^{-1}\\
   \end{array}
 \right).
\end{equation}
\hrule
\end{figure*}
\begin{figure*}
\begin{equation}\label{Eqn_det_Ba_C2}
\det(B)=\det\left(\begin{array}{cc}
  \det\left(\begin{array}{ccc}
      1  & 1   &  1     \\
     \lambda_{j_0}  & \lambda_{j_1}   & \lambda_{j_2}    \\
     \lambda_{j_0}^2 &\lambda_{j_1}^2 &\lambda_{j_2}^2\\
  \end{array}\right)   & \det\left(\begin{array}{ccc}
      1  & 1   &  1     \\
      \lambda_{t_0} & \lambda_{t_1} & \lambda_{t_2}    \\
     \lambda_{t_0}^2 &\lambda_{t_1}^2 &\lambda_{t_2}^2\\
  \end{array}\right) \\
  \det\left(\begin{array}{ccc}
     1  & 1   &  1     \\
     \lambda_{j_0}  & \lambda_{j_1}   & \lambda_{j_2}    \\
     \theta_i\lambda_{j_0}^{-1} &\theta_i\lambda_{j_1}^{-1} &\theta_i\lambda_{j_2}^{-1} \\
  \end{array}\right)   &  \det\left(\begin{array}{ccc}
       1  & 1   &  1     \\
      \lambda_{t_0} & \lambda_{t_1} & \lambda_{t_2}    \\
     \theta_k\lambda_{t_0}^{-1} &\theta_k\lambda_{t_1}^{-1} &\theta_k\lambda_{t_2}^{-1}\\
  \end{array}\right)
\end{array}\right).
\end{equation}
\hrule
\end{figure*}

\begin{Theorem}\label{Thm-Con3}
Let $q$ be a prime power such that there exists a multiplicative subgroup $G$ of $\mathbf{F}_q\backslash\{0\}$ of size at least $n$ and with at least $\mu$ cosets. Choosing $\lambda_{i,t}$, $i\in [0, n)$, $t\in [0, r)$ from $G$, then the code in Construction \ref{Con-PMDS2} is a PMDS array code if $\theta_0,\ldots,\theta_{\mu-1}$ are elements from distinct cosets of $G$.
\end{Theorem}

\begin{proof}
By \eqref{Eqn-PC-A} and the fact that $\lambda_0,\ldots,\lambda_{n-1}$ are pairwise distinct nonzero elements in $\mathbf{F}_q$, we easily have that $A_i-A_j$ and therefore $\begin{pmatrix} 
I&I\\
A_i&A_j
\end{pmatrix}
$ are nonsingular for any $i,j\in [0,n)$ with $i\ne j$, therefore, i) of Definition \ref{Def-PMDSA} is satisfied.

Now let us check ii) of Definition \ref{Def-PMDSA}. 
Suppose there are two failed nodes in every local group and two more anywhere. Similar to the proof of Theorem \ref{Eqn_Thm_C1}, we only need to analyze the following two cases:
\begin{itemize}
    \item Both the two extra failed nodes are in the same local group, say group $i$. Assume that nodes $in+j_0,\ldots,in+j_{3}$ are failed, where $0\le j_0<\cdots<j_{3}< n$. Let $J=\{j_0,\ldots,j_{3}\}$, then the original file can be reconstructed if the matrix
$
\left(\begin{array}{c}
     H_i[J]  \\
    P_i[J]
\end{array}\right)
$
is of full rank.

Note that
\begin{align*}
&\left(\begin{array}{c}
     H_i[J]  \\
    P_i[J]
\end{array}\right)\\=& \left(
\begin{array}{cccc}
      I & I & I & I  \\
    A_{j_0} & A_{j_1} & A_{j_2} & A_{j_{3}}\\
    \lambda_{j_0}^{2} I & \lambda_{j_1}^{2} I & \lambda_{j_2}^{2} I & \lambda_{j_3}^{2} I\\
   \theta_i \lambda_{j_0}^{-1} I & \theta_i\lambda_{j_1}^{-1} I & \theta_i\lambda_{j_2}^{-1} I &\theta_i \lambda_{j_3}^{-1} I
\end{array}
\right)\\
 =&\begin{psmallmatrix}
     1&0 &1&0 &1&0 &1&0 \\
     0&1&0 &1&0 &1&0 &1\\
      \lambda_{j_0}&x_0 &\lambda_{j_1}&x_1 &\lambda_{j_2}&x_2 &\lambda_{j_3}&x_3\\
     0&\lambda_{j_0}&0 &\lambda_{j_1}&0 &\lambda_{j_2}&0 &\lambda_{j_3}\\
          \lambda_{j_0}^2&0 &\lambda_{j_1}^2&0 &\lambda_{j_2}^2&0 &\lambda_{j_3}^2&0 \\
          0&\lambda_{j_0}^2&0 &\lambda_{j_1}^2&0 &\lambda_{j_2}^2&0 &\lambda_{j_3}^2\\
 \theta_i\lambda_{j_0}^{-1}&0 &\theta_i\lambda_{j_1}^{-1}&0 &\theta_i\lambda_{j_2}^{-1}&0 &\theta_i\lambda_{j_3}^{-1}&0 \\
          0&\theta_i\lambda_{j_0}^{-1}&0 &\theta_i\lambda_{j_1}^{-1}&0 &\theta_i\lambda_{j_2}^{-1}&0 &\theta_i\lambda_{j_3}^{-1}
          \\
   \end{psmallmatrix},\\
\end{align*}
where $x_j=0, 1$.
By swapping the rows and columns, it is easy to see that $\left(\begin{array}{c}
     H_i[J]  \\
    P_i[J]
\end{array}\right)$ is equivalent to the following matrix
\begin{equation}\label{Eqn-M-C2a}
\begin{psmallmatrix}
     1 & 1 & 1 & 1  & 0 & 0 & 0 & 0\\
      \lambda_{j_0}  &\lambda_{j_1} &\lambda_{j_2} &\lambda_{j_3}  &x_0 &x_1 &x_2 &x_3\\
          \lambda_{j_0}^2 &\lambda_{j_1}^2 &\lambda_{j_2}^2 &\lambda_{j_3}^2 & 0 & 0 & 0 & 0\\
 \theta_i\lambda_{j_0}^{-1} &\theta_i\lambda_{j_1}^{-1} &\theta_i\lambda_{j_2}^{-1} &\theta_i\lambda_{j_3}^{-1} & 0 & 0 & 0 & 0\\
 0& 0& 0 & 0 &1 &1 &1 &1\\
      0& 0& 0 & 0 & \lambda_{j_0} &\lambda_{j_1} &\lambda_{j_2} &\lambda_{j_3}\\
       0& 0& 0 & 0 & \lambda_{j_0}^2 &\lambda_{j_1}^2 &\lambda_{j_2}^2 &\lambda_{j_3}^2\\
          0& 0& 0 & 0 & \theta_i\lambda_{j_0}^{-1} &\theta_i\lambda_{j_1}^{-1} &\theta_i\lambda_{j_2}^{-1} &\theta_i\lambda_{j_3}^{-1}
    \end{psmallmatrix},
\end{equation}
under elementary transformation.
The matrix in \eqref{Eqn-M-C2a} is nonsinglualr if
\begin{equation}\label{Eqn-M-C2b}
 \Lambda= \left(\begin{array}{cccc}
         1 &1 &1 &1\\
  \lambda_{j_0} &\lambda_{j_1} &\lambda_{j_2} &\lambda_{j_3}\\
      \lambda_{j_0}^2 &\lambda_{j_1}^2 &\lambda_{j_2}^2 &\lambda_{j_3}^2 \\
         \theta_i\lambda_{j_0}^{-1} &\theta_i\lambda_{j_1}^{-1} &\theta_i\lambda_{j_2}^{-1} &\theta_i\lambda_{j_3}^{-1}
    \end{array}\right)
\end{equation}
is nonsingular.
By $\theta_i\ne 0$, we have
\begin{align*}
&\left(\begin{array}{cccc}
          & & &\theta_i^{-1}\\
  1 & & &\\
       &1 & &\\
          & & 1 &
    \end{array}\right)\Lambda\left(\begin{array}{cccc}
  \lambda_{j_0}  & & &\\
       &\lambda_{j_1}  & &\\
          & &\lambda_{j_2}  &\\
           & & &\lambda_{j_3}
    \end{array}\right)\\
    =&\left(\begin{array}{cccc}
         1 &1 &1 &1\\
  \lambda_{j_0} &\lambda_{j_1} &\lambda_{j_2} &\lambda_{j_3}\\
      \lambda_{j_0}^2 &\lambda_{j_1}^2 &\lambda_{j_2}^2 &\lambda_{j_3}^2 \\
         \lambda_{j_0}^{3} &\lambda_{j_1}^{3} &\lambda_{j_2}^{3} &\lambda_{j_3}^{3}
    \end{array}\right),
\end{align*}
which is the transpose of a Vandermonde matrix and has the same rank as the matrix $\Lambda$ in \eqref{Eqn-M-C2b}, in conjunction with the facts that $\lambda_{j_0}, \lambda_{j_1}, \lambda_{j_2}, \lambda_{j_3}$ are pairwise distinct nonzero elements, we conclude that the matrix $
\left(\begin{array}{c}
     H_i[J]  \\
    P_i[J]
\end{array}\right)
$ is nonsingular.

\item

Suppose that the two extra failed nodes are in two different local groups, say group $i$ and group $k$, where $0\le i<k<\mu$. Assume that nodes $in+j_0, in+j_1, in+j_{2}$ and nodes $kn+t_0, kn+t_1, kn+t_{2}$ are failed, where $0\le j_0<j_1<j_{2}< n$ and $0\le t_0<t_1<t_{2}< n$. Let $J=\{j_0,j_1,j_{2}\}$ and $T=\{t_0, t_1, t_{2}\}$, then the original file can be reconstructed if the following matrix
\begin{align*}
\hat{\mathcal{H}}=&\begin{pmatrix}
H_i[J]&\\
& H_{k}[T]\\
P_i[J]& P_{k}[T]
\end{pmatrix}\\
=&
\begin{psmallmatrix}
      I & I & I &&& \\
    A_{j_0} & A_{j_1} & A_{j_{2}}&&&\\
 &&&  I & I  & I   \\
  &&&  A_{t_0} & A_{t_1} &  A_{t_{2}}\\
   \lambda_{j_0}^{2}I & \lambda_{j_1}^{2} I  & \lambda_{j_{2}}^{2} I &  \lambda_{t_0}^{2} I & \lambda_{t_1}^{2} I & \lambda_{t_{2}}^{2} I\\
    \theta_i \lambda_{j_0}^{-1}I &\theta_i \lambda_{j_1}^{-1}I &\theta_i \lambda_{j_{2}}^{-1}I & \theta_{k} \lambda_{t_0}^{-1} I &\theta_{k} \lambda_{t_1}^{-1} I &\theta_{k} \lambda_{t_{2}}^{-1} I
\end{psmallmatrix}
\end{align*}
is nonsingular. Substituting \eqref{Eqn-PC-A} into $\hat{\mathcal{H}}$, we have \eqref{Eqn_hatH_C2},
where $x_0,\ldots,x_5=0$ or $1$. By swapping the rows and columns of $\hat{\mathcal{H}}$ in \eqref{Eqn_hatH_C2}, it is easy to see that $\hat{\mathcal{H}}$ is equivalent to the matrix in \eqref{Eqn_hatH'_C3}
under elementary transformation,
which is nonsingular if the following matrix
\begin{equation*}
B\hspace{-1mm}=\hspace{-1.5mm}\begin{pmatrix}\hspace{-1.3mm}
        1  & 1   &  1  &0& 0& 0    \\
     \lambda_{j_0}  & \lambda_{j_1}   & \lambda_{j_2}  &0& 0& 0  \\
    0& 0& 0 & 1 & 1  &  1   \\
  0& 0& 0 & \lambda_{t_0} & \lambda_{t_1} & \lambda_{t_2}\\
 \lambda_{j_0}^2 &\lambda_{j_1}^2 &\lambda_{j_2}^2 &\lambda_{t_0}^2 &\lambda_{t_1}^2 &\lambda_{t_2}^2\\
  \theta_i\lambda_{j_0}^{-1} &\theta_i\lambda_{j_1}^{-1} &\theta_i\lambda_{j_2}^{-1} &\theta_k\lambda_{t_0}^{-1} &\theta_k\lambda_{t_1}^{-1} &\theta_k\lambda_{t_2}^{-1}\\
\hspace{-1.9mm}\end{pmatrix}
\end{equation*}
is nonsingular.

By Lemma \ref{Lemma-Det}, we have \eqref{Eqn_det_Ba_C2}.
By factoring out the nonzero Vandermonde determinant from each column of the determinant in \eqref{Eqn_det_Ba_C2}, we further have
\begin{equation*}
\det(B)\ne 0\iff \det\left(
\begin{array}{cc}
  \lambda_{j_0}\lambda_{j_1}\lambda_{j_2} & \lambda_{t_0}\lambda_{t_1}\lambda_{t_2}    \\
  \theta_i  & \theta_k
\end{array}
\right)\ne 0,
\end{equation*}
where the last inequality holds since both $\lambda_{j_0}\lambda_{j_1}\lambda_{j_2}$ and $\lambda_{t_0}\lambda_{t_1}\lambda_{t_2}$ are in the subgroup $G$, while  $\theta_i$ and $\theta_{k}$ are in different cosets of $G$.

Therefore, $B$ and thus $\hat{\mathcal{H}}$ is nonsingular, and the original file can be reconstructed.
\end{itemize}
This finishes the proof.
\end{proof}

\begin{Remark}
Although Construction \ref{Con-PMDS2} only supports two local parities, it has the following novelty.
\begin{itemize}
  \item  It is the first time to use non-diagonal matrices (cf. \eqref{Eqn-PC-A}) as building blocks in PMDS array codes besides diagonal matrices, i.e., it is not required that each sub-stripe of the local code be a scalar MDS code as in \cite{holzbaur2021partial}, which is the key to  reduce the rebuilding access.
  \item  For the local code of the new PMDS array code in Construction \ref{Con-PMDS2}, its rebuilding access is around $0.75$ (more precisely, $\frac{3n/2-2}{2n-4}$) times that of a Reed-Solomon code with the same parameters. In \cite{wu2021achievable}, a tight lower bound of the average rebuilding access of $[n, n-2]$ MDS array codes with sub-packetization level $2$ as well as the optimal code construction was derived. Although the exact expression of the general lower bound is complicated, it was shown that the  average rebuilding access of $[n, n-2]$ MDS array codes with sub-packetization level $2$ is larger than $0.72$ times that of an RS code with the same parameters for $n\le 50$. This shows that the local code of the new PMDS array code in Construction \ref{Con-PMDS2} that we choose has rebuilding access which is about $1.04$ times the lower bound in \cite{wu2021achievable} for $n\le 50$.

Of course, we can choose an MDS array code with the optimal average rebuilding access in \cite{wu2021achievable} as the local code. However, it will be very difficult to verify ii) of Definition \ref{Def-PMDSA}, which will be left as our future research.
\end{itemize}
\end{Remark}

\section{A new $(\mu, n; r, s=3)$ PMDS array code construction}\label{sec:code_r=3}

In this subsection, we propose a new $(\mu, n; r, s=3)$ PMDS array code construction, where each $P_i$ in \eqref{Eqn-PC-PMDS} has three block rows. To prove that an array code defined by the parity-check matrix in \eqref{Eqn-PC-PMDS} is a PMDS array code, similar to the previous sections, we need to calculate the determinants of
the sub-block matrices that are formed by any $r+3$ thick columns of $\begin{pmatrix}
H_i\\
P_i
\end{pmatrix}
$, any $r+2$ thick columns of $\begin{pmatrix}
H_i\\
P_i^{[T]}
\end{pmatrix}$ ($T\subset [0, 3)$, $|T|=2$), and any $r+1$ thick columns of $\begin{pmatrix}
H_i\\
P_i^{[t]}
\end{pmatrix}
$ ($t=0,1,2$). If we define $P_i$ similar to that in \eqref{Eqn-PC-new3}, e.g., by adding a block row 
\begin{equation*}
\left(\begin{array}{cccc}
    A_0^{r+1} & A_1^{r+1} & \cdots & A_{n-1}^{r+1}
\end{array}\right)
\end{equation*}
 or 
 \begin{equation*}
 \left(\begin{array}{cccc}
    \theta'_{i}A_0^{-2} & \theta'_{i}A_1^{-2} & \cdots & \theta'_{i}A_{n-1}^{-2}
\end{array}\right),
\end{equation*} then some of the sub-block matrices mentioned above will not be equivalent to a block Vandermonde matrix anymore, and their determinants will be hard to calculate. 
By defining $P_i$ as $\ell$ matrices of order $3\times \ell$,  with each $3\times \ell$ matrix being a sub-matrix of a Cauchy-Vandermonde matrix, the above concern can be addressed by applying Lemma \ref{Le_det_CV}.

Following the notation in Construction \ref{Con-C5},
we define two variants of the matrices $A_i$ ($i\in [0, n)$) in the following. For $i\in [0, n)$, let $A'_i$ and $A''_i$ be $\ell\times \ell$ matrices that satisfy \begin{equation}\label{Eqn Hadamard coding matrix-N1}
\left(
                     \begin{array}{c}
                       V_{\overline{i},0} \\
                       V_{\overline{i},1} \\
                       \vdots \\
                       V_{\overline{i},r-1}
                     \end{array}
                   \right)A'_i=\left(
                     \begin{array}{c}
                    \frac{1}{\lambda_{i,0}-d_0}   V_{\overline{i},0} \\
                      \frac{1}{\lambda_{i,1}-d_0}  V_{\overline{i},1} \\
                       \vdots \\
                     \frac{1}{\lambda_{i,r-1}-d_0}  V_{\overline{i},r-1}
                     \end{array}
                   \right),
\end{equation}
and
\begin{equation}\label{Eqn Hadamard coding matrix-N2}
\left(
                     \begin{array}{c}
                       V_{\overline{i},0} \\
                       V_{\overline{i},1} \\
                       \vdots \\
                       V_{\overline{i},r-1}
                     \end{array}
                   \right)A''_i=\left(
                     \begin{array}{c}
                    \frac{1}{\lambda_{i,0}-d_1}   V_{\overline{i},0} \\
                      \frac{1}{\lambda_{i,1}-d_1}  V_{\overline{i},1} \\
                       \vdots \\
                     \frac{1}{\lambda_{i,r-1}-d_1}  V_{\overline{i},r-1}
                     \end{array}
                   \right),
\end{equation}
where $d_0, d_1 \in \mathbf{F}_q\backslash\{\lambda_{i,t}, i\in [0, n), t\in [0, r)\}$ and $d_0\ne d_1$.

Similarly, we have that $A'_i$ and $A''_i$ are diagonal matrices and the $a$-th rows of $A'_i$ and $A''_i$ are
\begin{equation}\label{Eqn-row-of-A'}
e_aA'_i=\frac{1}{\lambda_{i, a_{\overline{i}}}-d_0}e_a ~\mbox{and}~ e_aA''_i=\frac{1}{\lambda_{i, a_{\overline{i}}}-d_1}e_a
\end{equation}
respectively,
where $i\in [0, n)$, $a\in [0, \ell)$, and $a_{\overline{i}}$ denotes the $\overline{i}$-th element in the $r$-ary expansion of $a$.

\begin{Construction}\label{Con-news3}
Let $\mu, r, n', n$ be four positive integers, where  $\mu, r\ge2$ and $n>n'$, and let $\ell=r^{n'}$. We construct a new $[\mu n, \mu (n-r)-3,\ell]$ array code over $\mathbf{F}_q$ with the parity-check matrix having the form as in \eqref{Eqn-PC-PMDS} with
\begin{equation}\label{Eqn-PC-con3-1}
H_i=H,
\end{equation}
and
\begin{equation}\label{Eqn-PC-con3-2}
P_i=\left(\begin{array}{cccc}
    A_0^{r} & A_1^{r} & \cdots & A_{n-1}^{r} \\
    \theta_{i}A'_0 & \theta_{i}A'_1& \cdots & \theta_{i}A'_{n-1} \\
    \delta_{i}A''_0 & \delta_{i}A''_1& \cdots & \delta_{i}A''_{n-1}
\end{array}\right)
\end{equation}
for $i\in [0, \mu)$,
where $H$ is defined in \eqref{Eqn-PC-C5}, i.e., the parity-check matrix of the code in Construction \ref{Con-C5},
$A_i$, $A'_i$, and $A''_i$ are $\ell\times \ell$ matrices defined in \eqref{Eqn Hadamard coding matrix}, \eqref{Eqn Hadamard coding matrix-N1}, and \eqref{Eqn Hadamard coding matrix-N2}, respectively, $\theta_i, \delta_i\in \mathbf{F}_q\backslash\{0\}$ for $i\in [0, \mu)$.
\end{Construction}

\begin{Theorem}\label{Eqn_Thm_C3}
The code in Construction \ref{Con-news3} is a $(\mu, n; r, s=3)$  PMDS array code over $\mathbf{F}_{q}$ with sub-packetization level $\ell$ if the following conditions C1--C4 hold.
\begin{itemize}
    \item [C1.] $q=q_0^3$, where $q_0>\mu(\Phi+1)$ is a prime power such that there exists a multiplicative subgroup $G$ of $\mathbf{F}_{q_0}\backslash\{0\}$ of size at least $\Phi+1$ and with at least $\mu$ cosets, where $\Phi$ is defined in \eqref{Eqn-Phi}.
\item [C2.] $d_0$ and $d_1$ are two distinct elements chosen from $\mathbf{F}_{q_0}$. $\lambda_{i,t}$, $i\in [0, n)$, $t\in [0, r)$ are chosen from $\Omega\backslash\{d_0\}$ such that R1-R3 of Lemma \ref{lem repair} hold, where
\begin{equation*}
\Omega=\{\lambda\in \mathbf{F}_{q_0}:\frac{1}{\lambda-d_1}\in G\},
\end{equation*}
and clearly $|\Omega|\ge \Phi+1$.
\item [C3.] $\Theta=\{\theta_0,\theta_1, \ldots, \theta_{\mu-1}\}\subset \mathbf{F}_q\backslash\{0\}$ is $3$--wise independent over $\mathbf{F}_{q_0}$, i.e., any $t$-subset of $\Theta$ with $t\le 3$ is linearly independent over $\mathbf{F}_{q_0}$.
\item [C4.] $\delta_0,\delta_1, \ldots, \delta_{\mu-1}\in \mathbf{F}_{q_0}$ are elements from distinct cosets of $G$.
\end{itemize}
In addition, the repair bandwidth and the rebuilding access of node $in+j$ are
\begin{equation*}
 \gamma_{in+j} \hspace{-1mm}=\hspace{-1.4mm}\left\{\hspace{-2.4mm}
                   \begin{array}{ll}
(1\hspace{-.4mm}+\hspace{-.4mm}\frac{(\lceil\frac{n}{n'}\rceil-1)(r-1)}{n-1})\frac{\ell}{r}(n-1), \hspace{-.9mm}&\mbox{if\ } 0\hspace{-.4mm}\le\hspace{-.4mm} j\% n' \hspace{-.6mm}< \hspace{-.6mm}n\%n', \\
                       (1\hspace{-.4mm}+\hspace{-.4mm}\frac{(\lfloor\frac{n}{n'}\rfloor-1)(r-1)}{n-1})\frac{\ell}{r}(n-1),\hspace{-1mm} & \mbox{otherwise},
                      \end{array}
                    \right.
\end{equation*}
and
\begin{equation*}
\Gamma_{in+j}=\ell(n-1)
\end{equation*}
for $i\in [0, \mu)$ and $j\in [0, n)$, respectively.
\end{Theorem}

\begin{proof}
Firstly, we introduce a method to choose the subset $\{\theta_0, \theta_1,\ldots, \theta_{\mu-1}\}$ in \cite{gopi2020maximally}. Let $v_0,v_1,v_2$ be a basis of $\mathbf{F}_q$ over $\mathbf{F}_{q_0}$, and let $\xi_0,\xi_1,\ldots,\xi_{\mu-1}$ be pairwise distinct elements in $\mathbf{F}_{q_0}$. Define
\begin{equation*}
\theta_i= v_0+\xi_i v_1+ \xi_i^2 v_2,~i\in [0, \mu),
\end{equation*}
then
$\{\theta_0, \theta_1,\ldots, \theta_{\mu-1}\}$
is $3$--wise independent over $\mathbf{F}_{q_0}$.

Secondly, we focus on the repair property of the new code. By \eqref{Eqn-PC-PMDS} and \eqref{Eqn-PC-con3-1}, we see that each local code is an MDS array code defined by the parity-check matrix $H$ of the code in Construction \ref{Con-C5}, i.e., i) of Definition \ref{Def-PMDSA} is satisfied,
and the result on the repair bandwidth and rebuilding access are derived directly according to Lemmas \ref{lem repair} and \ref{lem-C5-q}.

Lastly, let us prove that ii) of Definition \ref{Def-PMDSA} is also satisfied.
Suppose that there are $r$ failed nodes in every local group and three more anywhere. Similarly, we only need to analyze the following three cases:
\begin{itemize}
    \item All of the three extra failed nodes are in the same local group, say group $i$ and assume that nodes $in+j_0,\ldots,in+j_{r+2}$ are failed, where $0\le j_0<\cdots<j_{r+2}< n$. Let $J=\{j_0,\ldots,j_{r+2}\}$, then the original file can be reconstructed if the following matrix
\begin{equation*}
\left(\begin{array}{c}
     H_i[J]  \\
    P_i[J]
\end{array}\right)=\left(
\begin{array}{cccc}
      I & I & \cdots & I  \\
    A_{j_0} & A_{j_1} & \cdots & A_{j_{r+2}}\\
    \vdots & \vdots & \ddots & \vdots\\
    A_{j_0}^{r-1} & A_{j_1}^{r-1} & \cdots & A_{j_{r+2}}^{r-1} \\
    A_{j_0}^{r} & A_{j_1}^{r} & \cdots & A_{j_{r+2}}^{r} \\
    \theta_i A'_{j_0} &\theta_i A'_{j_1} & \cdots &\theta_i A'_{j_{r+2}}\\
    \delta_i A''_{j_0} &\delta_i A''_{j_1} & \cdots &\delta_i A''_{j_{r+2}}\\
\end{array}
\right)
\end{equation*}
is of full rank. Clearly, each block entry in $\left(\begin{array}{c}
     H_i[J]  \\
    P_i[J]
\end{array}\right)$ is an $\ell\times \ell$ diagonal matrix, by swapping the rows and columns, $\left(\begin{array}{c}
     H_i[J]  \\
    P_i[J]
\end{array}\right)$ is equivalent to
\begin{equation*}
\mbox{blkdiag}(B_0, B_1, \ldots, B_{\ell-1})
\end{equation*}
under elementary transformation, which has the same rank as $\left(\begin{array}{c}
     H_i[J]  \\
    P_i[J]
\end{array}\right)$, where $B_a$ is formed by the $a, a+\ell, \ldots, (a+(r+2)\ell)$-th rows and the $a, a+\ell, \ldots, (a+(r+2)\ell)$-th columns of $\left(\begin{array}{c}
     H_i[J]  \\
    P_i[J]
\end{array}\right)$ for $a\in [0, \ell)$, i.e.,
\begin{equation*}
B_a=
\begin{pmatrix}
      1 & 1 & \cdots & 1  \\
    \lambda_{j_0,a_{\overline{j_0}}} &  \lambda_{j_1,a_{\overline{j_1}}} & \cdots & \lambda_{j_{r+2},a_{\overline{j_{r+2}}}} \\
    \vdots & \vdots & \ddots & \vdots\\
     \lambda_{j_0,a_{\overline{j_0}}}^{r-1} &  \lambda_{j_1,a_{\overline{j_1}}}^{r-1} & \cdots & \lambda_{j_{r+2},a_{\overline{j_{r+2}}}}^{r-1} \\
   \lambda_{j_0,a_{\overline{j_0}}}^{r} &  \lambda_{j_1,a_{\overline{j_1}}}^{r} & \cdots & \lambda_{j_{r+2},a_{\overline{j_{r+2}}}}^{r} \\  \frac{\theta_i}{\lambda_{j_0,a_{\overline{j_0}}}-d_0} &  \frac{\theta_i}{\lambda_{j_1,a_{\overline{j_1}}}-d_0}  & \cdots & \frac{\theta_i}{\lambda_{j_{r+2},a_{\overline{j_{r+2}}}}-d_0}
    \\
 \frac{\delta_i}{\lambda_{j_0,a_{\overline{j_0}}}-d_1} &  \frac{\delta_i}{\lambda_{j_1,a_{\overline{j_1}}}-d_1}  & \cdots & \frac{\delta_i}{\lambda_{j_{r+2},a_{\overline{j_{r+2}}}}-d_1}
\end{pmatrix}
\end{equation*}
by \eqref{Eqn-row-of-A} and \eqref{Eqn-row-of-A'}.
which is a transpose of a Cauchy-Vandermonde matrix (after scaling and permuting rows) and thus is nonsingular by Lemma \ref{Le_det_CV} and C2. Therefore, $\left(\begin{array}{c}
     H_i[J]  \\
    P_i[J]
\end{array}\right)$ is nonsingular and the original file can be reconstructed.

\item
The three extra node failures are in two different groups, say one in group $i$ and two in group $k$, and assume that nodes $in+j_0,\ldots,in+j_{r}$ and nodes $kn+t_0,\ldots,kn+t_{r+1}$ are failed, where $0\le i, k<\mu$, $0\le j_0<\cdots<j_{r}< n$ and $0\le t_0<\cdots<t_{r+1}< n$. Let $J=\{j_0,\ldots,j_{r}\}$ and $T=\{t_0,\ldots,t_{r+1}\}$, then the original file can be reconstructed if the following matrix
\begin{align*}
\hat{\mathcal{H}}=&\begin{pmatrix}
H_i[J]&\\
& H_{k}[T]\\
P_i[J]& P_{k}[T]
\end{pmatrix}\\=&
\begin{psmallmatrix}
      I & I & \cdots & I  &&&& \\
    A_{j_0} & A_{j_1} & \cdots & A_{j_{r}}&&&&\\
    \vdots & \vdots & \ddots & \vdots&&&&\\
    A_{j_0}^{r-1} & A_{j_1}^{r-1} & \cdots & A_{j_{r}}^{r-1} &&&&\\
  &&&&  I & I & \cdots & I   \\
  &&&&  A_{t_0} & A_{t_1} & \cdots & A_{t_{r+1}}\\
  &&&&  \vdots & \vdots & \ddots & \vdots\\
   &&&& A_{t_0}^{r-1} & A_{t_1}^{r-1} & \cdots & A_{t_{r+1}}^{r-1} \\
   A_{j_0}^{r} & A_{j_1}^{r} & \cdots & A_{j_{r}}^{r} &  A_{t_0}^{r} & A_{t_1}^{r} & \cdots & A_{t_{r+1}}^{r} \\
     \theta_i A'_{j_0} &\theta_i A'_{j_1} & \cdots &\theta_i A'_{j_{r}} & \theta_k A'_{t_0} &\theta_k A'_{t_1} & \cdots &\theta_k A'_{t_{r+1}}\\
    \delta_i A''_{j_0} &\delta_i A''_{j_1} & \cdots &\delta_i A''_{j_{r}} & \delta_k A''_{t_0} &\delta_k A''_{t_1} & \cdots &\delta_k A''_{t_{r+1}}\\
\end{psmallmatrix}
\end{align*}
is nonsingular.

Similarly, by swapping the rows and columns of $\hat{\mathcal{H}}$, $\hat{\mathcal{H}}$ is equivalent to
\begin{equation*}
\mbox{blkdiag}(B_0, B_1, \ldots, B_{\ell-1})
\end{equation*}
under elementary transformation, which has the same rank as $\hat{\mathcal{H}}$, where $B_a$ is formed by the $a, a+\ell, \ldots, (a+(2r+2)\ell$)-th rows and the $a, a+\ell, \ldots, (a+(2r+2)\ell)$-th columns of $\hat{\mathcal{H}}$, and can be expressed as in \eqref{Eqn_Ba_C3} according to \eqref{Eqn-row-of-A}.
\begin{figure*}
\begin{equation}\label{Eqn_Ba_C3}
B_a=\begin{pmatrix}
      1 & 1 & \cdots & 1  &&&& \\
    \lambda_{j_0,a_{\overline{j_0}}} &  \lambda_{j_1,a_{\overline{j_1}}} & \cdots & \lambda_{j_r,a_{\overline{j_r}}} &&&&\\
    \vdots & \vdots & \ddots & \vdots&&&&\\
     \lambda_{j_0,a_{\overline{j_0}}}^{r-1} &  \lambda_{j_1,a_{\overline{j_1}}}^{r-1} & \cdots & \lambda_{j_r,a_{\overline{j_r}}}^{r-1} &&&&\\
  &&&&  1 & 1 & \cdots & 1   \\
  &&&&  \lambda_{t_0,a_{\overline{t_0}}} &  \lambda_{t_1,a_{\overline{t_1}}} & \cdots & \lambda_{t_{r+1},a_{\overline{t_{r+1}}}}\\
  &&&&  \vdots & \vdots & \ddots & \vdots\\
   &&&& \lambda_{t_0,a_{\overline{t_0}}}^{r-1} &  \lambda_{t_1,a_{\overline{t_1}}}^{r-1} & \cdots & \lambda_{t_{r+1},a_{\overline{t_{r+1}}}}^{r-1}\\
   \lambda_{j_0,a_{\overline{j_0}}}^{r} &  \lambda_{j_1,a_{\overline{j_1}}}^{r} & \cdots & \lambda_{j_r,a_{\overline{j_r}}}^{r} &  \lambda_{t_0,a_{\overline{t_0}}}^{r} &  \lambda_{t_1,a_{\overline{t_1}}}^{r} & \cdots & \lambda_{t_{r+1},a_{\overline{t_{r+1}}}}^{r} \\
    \frac{\theta_i}{\lambda_{j_0,a_{\overline{j_0}}}-d_0}     &   \frac{\theta_i}{\lambda_{j_1,a_{\overline{j_1}}}-d_0} & \cdots &  \frac{\theta_i}{\lambda_{j_r,a_{\overline{j_r}}}-d_0}&  \frac{\theta_{k}}{\lambda_{t_0,a_{\overline{t_0}}}-d_0} &  \frac{\theta_{k}}{\lambda_{t_1,a_{\overline{t_1}}}-d_0} & \cdots & \frac{\theta_{k}}{\lambda_{t_{r+1},a_{\overline{t_{r+1}}}}-d_0}\\
    \frac{\delta_i}{\lambda_{j_0,a_{\overline{j_0}}}-d_1}     &   \frac{\delta_i}{\lambda_{j_1,a_{\overline{j_1}}}-d_1} & \cdots &  \frac{\delta_i}{\lambda_{j_r,a_{\overline{j_r}}}-d_1}&  \frac{\delta_{k}}{\lambda_{t_0,a_{\overline{t_0}}}-d_1} &  \frac{\delta_{k}}{\lambda_{t_1,a_{\overline{t_1}}}-d_1} & \cdots & \frac{\delta_{k}}{\lambda_{t_{r+1},a_{\overline{t_{r+1}}}}-d_1}
\end{pmatrix}, ~a\in [0, \ell).
\end{equation}
\hrule
 \end{figure*}

By Lemma \ref{Lemma-Det2},
for $a\in [0, \ell)$,  $\det(B_a)\ne 0$ is equivalent to the inequality in \eqref{Eqn_Ka_C3} in the next page.
\begin{figure*}
\begin{align}\label{Eqn_Ka_C3}
\nonumber K_a=&\det\left(\begin{array}{cccc}
             1 & 1 & \cdots & 1   \\
    \lambda_{j_0,a_{\overline{j_0}}} &  \lambda_{j_1,a_{\overline{j_1}}} & \cdots & \lambda_{j_r,a_{\overline{j_r}}} \\
    \vdots & \vdots & \ddots & \vdots\\
     \lambda_{j_0,a_{\overline{j_0}}}^{r} &  \lambda_{j_1,a_{\overline{j_1}}}^{r} & \cdots & \lambda_{j_r,a_{\overline{j_r}}}^{r} \\
         \end{array}\right)\det\left(\begin{array}{cccc}
              1 & 1 & \cdots & 1   \\
    \lambda_{t_0,a_{\overline{t_0}}} &  \lambda_{t_1,a_{\overline{t_1}}} & \cdots & \lambda_{t_{r+1},a_{\overline{t_{r+1}}}}\\
    \vdots & \vdots & \ddots & \vdots\\
    \lambda_{t_0,a_{\overline{t_0}}}^{r-1} &  \lambda_{t_1,a_{\overline{t_1}}}^{r-1} & \cdots & \lambda_{t_{r+1},a_{\overline{t_{r+1}}}}^{r-1}\\
      \frac{\theta_{k}}{\lambda_{t_0,a_{\overline{t_0}}}-d_0} &  \frac{\theta_{k}}{\lambda_{t_1,a_{\overline{t_1}}}-d_0} & \cdots & \frac{\theta_{k}}{\lambda_{t_{r+1},a_{\overline{t_{r+1}}}}-d_0}\\
  \frac{\delta_{k}}{\lambda_{t_0,a_{\overline{t_0}}}-d_1} &  \frac{\delta_{k}}{\lambda_{t_1,a_{\overline{t_1}}}-d_1} & \cdots & \frac{\delta_{k}}{\lambda_{t_{r+1},a_{\overline{t_{r+1}}}}-d_1}
         \end{array}\right)\\
\nonumber  &- \theta_i\delta_{k}\det\underbrace {\left(\hspace{-1mm}\begin{array}{cccc}
             1 & 1 & \cdots & 1   \\
    \lambda_{j_0,a_{\overline{j_0}}} &  \lambda_{j_1,a_{\overline{j_1}}} & \cdots & \lambda_{j_r,a_{\overline{j_r}}} \\
    \vdots & \vdots & \ddots & \vdots\\
     \lambda_{j_0,a_{\overline{j_0}}}^{r-1} &  \lambda_{j_1,a_{\overline{j_1}}}^{r-1} & \cdots & \lambda_{j_r,a_{\overline{j_r}}}^{r-1} \\
    \frac{1}{\lambda_{j_0,a_{\overline{j_0}}}-d_0}     &   \frac{1}{\lambda_{j_1,a_{\overline{j_1}}}-d_0} & \cdots &  \frac{1}{\lambda_{j_r,a_{\overline{j_r}}}-d_0}
         \end{array}\hspace{-1mm}\right)}_{A} \det\underbrace{\left(\hspace{-1mm}\begin{array}{cccc}
              1 & 1 & \cdots & 1   \\
    \lambda_{t_0,a_{\overline{t_0}}} &  \lambda_{t_1,a_{\overline{t_1}}} & \cdots & \lambda_{t_r,a_{\overline{t_r}}}\\
    \vdots & \vdots & \ddots & \vdots\\
  \lambda_{t_0,a_{\overline{t_0}}}^{r} &  \lambda_{t_1,a_{\overline{t_1}}}^{r} & \cdots & \lambda_{t_{r+1},a_{\overline{t_{r+1}}}}^{r} \\
      \frac{1}{\lambda_{t_0,a_{\overline{t_0}}}-d_1} &  \frac{1}{\lambda_{t_1,a_{\overline{t_1}}}-d_1} & \cdots & \frac{1}{\lambda_{t_{r+1},a_{\overline{t_{r+1}}}}-d_1}
         \end{array}\hspace{-1mm}\right)}_{B}\\
\nonumber &+ \det\left(\begin{array}{cccc}
             1 & 1 & \cdots & 1   \\
    \lambda_{j_0,a_{\overline{j_0}}} &  \lambda_{j_1,a_{\overline{j_1}}} & \cdots & \lambda_{j_r,a_{\overline{j_r}}} \\
    \vdots & \vdots & \ddots & \vdots\\
     \lambda_{j_0,a_{\overline{j_0}}}^{r-1} &  \lambda_{j_1,a_{\overline{j_1}}}^{r-1} & \cdots & \lambda_{j_r,a_{\overline{j_r}}}^{r-1} \\
   \frac{\delta_i}{\lambda_{j_0,a_{\overline{j_0}}}-d_1}     &   \frac{\delta_i}{\lambda_{j_1,a_{\overline{j_1}}}-d_1} & \cdots &  \frac{\delta_i}{\lambda_{j_r,a_{\overline{j_r}}}-d_1}
         \end{array}\right) \det\left(\begin{array}{cccc}
              1 & 1 & \cdots & 1   \\
    \lambda_{t_0,a_{\overline{t_0}}} &  \lambda_{t_1,a_{\overline{t_1}}} & \cdots & \lambda_{t_r,a_{\overline{t_r}}}\\
    \vdots & \vdots & \ddots & \vdots\\
  \lambda_{t_0,a_{\overline{t_0}}}^{r} &  \lambda_{t_1,a_{\overline{t_1}}}^{r} & \cdots & \lambda_{t_{r+1},a_{\overline{t_{r+1}}}}^{r} \\
     \frac{\theta_{k}}{\lambda_{t_0,a_{\overline{t_0}}}-d_0} &  \frac{\theta_{k}}{\lambda_{t_1,a_{\overline{t_1}}}-d_0} & \cdots & \frac{\theta_{k}}{\lambda_{t_{r+1},a_{\overline{t_{r+1}}}}-d_0}
         \end{array}\right)\\
         \ne &0.
\end{align}
\hrule
 \end{figure*}
Note that $\det(K_a)$ is an $\mathbf{F}_{q_0}$-linear combination of $\theta_i$ and $\theta_k$, the coefficient of $\theta_i$ arises from the second term, which is a nonzero element in $\mathbf{F}_{q_0}$ by Lemma \ref{Le_det_CV}, C2 and C4 because the matrices $A$ and $B$ in the second term are Cauchy-Vandermonde matrices (after permuting rows and transposing). By C3, this linear combination cannot be zero.
Therefore, $\mbox{blkdiag}(B_0, B_1, \ldots, B_{\ell-1})$ and thus $\hat{\mathcal{H}}$ is nonsingular, and the original file can be reconstructed.

\item The three extra node failures are in three distinct groups, say in groups $i$, $k$, and $l$, and assume that nodes $in+j_0,\ldots,in+j_{r}$,  nodes $kn+t_0,\ldots,kn+t_{r}$, and nodes $ln+u_0,\ldots,ln+u_{r}$  are failed, where $0\le i<k<l<\mu$, $0\le j_0<\cdots<j_{r}< n$, $0\le t_0<\cdots<t_{r}< n$,  and $0\le u_0<\cdots<u_{r}< n$. Let $J=\{j_0,\ldots,j_{r}\}$, $T=\{t_0,\ldots,t_{r}\}$, and $U=\{u_0,\ldots,u_{r}\}$, then the original file can be reconstructed if the matrix $\hat{\mathcal{H}}$ in \eqref{Eqn_hatH_C3} is nonsingular.
\begin{figure*}
\begin{align}\label{Eqn_hatH_C3}
\nonumber \hat{\mathcal{H}}&=\begin{pmatrix}
H_i[J]& &\\
& H_{k}[T]&\\
&& H_{l}[U]&\\
P_i[J]& P_{k}[T]& P_{l}[U]
\end{pmatrix}\\&=
\left(\begin{array}{cccccccccccc}
  I & I & \cdots & I  &&&&&&&& \\
    A_{j_0} & A_{j_1} & \cdots & A_{j_{r}}&&&&&&&&\\
    \vdots & \vdots & \ddots & \vdots&&&&&&&&\\
    A_{j_0}^{r-1} & A_{j_1}^{r-1} & \cdots & A_{j_{r}}^{r-1} &&&&&&&&\\
  &&&&  I & I & \cdots & I  &&&& \\
  &&&&  A_{t_0} & A_{t_1} & \cdots & A_{t_{r}}&&&&\\
  &&&&  \vdots & \vdots & \ddots & \vdots&&&&\\
   &&&& A_{t_0}^{r-1} & A_{t_1}^{r-1} & \cdots & A_{t_{r}}^{r-1}&&&& \\
   &&&& &&&&  I & I & \cdots & I  \\
&&&& &&&&  A_{u_0} & A_{u_1} & \cdots & A_{u_{r}}\\
&&&&  &&&&  \vdots & \vdots & \ddots & \vdots\\
&&&&   &&&& A_{u_0}^{r-1} & A_{u_1}^{r-1} & \cdots & A_{u_{r}}^{r-1} \\
   A_{j_0}^{r} & A_{j_1}^{r} & \cdots & A_{j_{r}}^{r} &  A_{t_0}^{r} & A_{t_1}^{r} & \cdots & A_{t_{r}}^{r} & A_{u_0}^{r} & A_{u_1}^{r} & \cdots & A_{u_{r}}^{r} \\
     \theta_i A'_{j_0} &\theta_i A'_{j_1} & \cdots &\theta_i A'_{j_{r}} & \theta_k A'_{t_0} &\theta_k A'_{t_1} & \cdots &\theta_k A'_{t_{r}}& \theta_l A'_{u_0} &\theta_l A'_{u_1} & \cdots &\theta_l A'_{u_{r}} \\
    \delta_i A''_{j_0} &\delta_i A''_{j_1} & \cdots &\delta_i A''_{j_{r}} & \delta_k A''_{t_0} &\delta_k A''_{t_1} & \cdots &\delta_k A''_{t_{r}}& \delta_k A''_{u_0} &\delta_k A''_{u_1} & \cdots &\delta_k A''_{u_{r}}
\end{array}\right).
\end{align}\hrule
 \end{figure*}

Similarly, by swapping the rows and columns of $\hat{\mathcal{H}}$, $\hat{\mathcal{H}}$ is equivalent to
\begin{equation*}
\mbox{blkdiag}(B_0, B_1, \ldots, B_{\ell-1})
\end{equation*}
under elementary transformation, where $B_a$ is formed by the $a, a+\ell, \ldots, (a+(3r+2)\ell)$-th rows and the $a, a+\ell, \ldots, (a+(3r+2)\ell)$-th columns of $\hat{\mathcal{H}}$ for $a\in [0, \ell)$, and can be expressed as in \eqref{Eqn_Ba_C3_3NF} according to \eqref{Eqn-row-of-A}.
\begin{figure*}
{\small
\begin{equation}\label{Eqn_Ba_C3_3NF}
B_a=\left(
\setlength{\arraycolsep}{1.8pt}
\begin{array}{cccccccccccc}
      1 & 1 & \cdots & 1  &&&&&&&& \\
    \lambda_{j_0,a_{\overline{j_0}}} &  \lambda_{j_1,a_{\overline{j_1}}} & \cdots & \lambda_{j_r,a_{\overline{j_r}}} &&&&&&&&\\
    \vdots & \vdots & \ddots & \vdots&&&&&&&&\\
     \lambda_{j_0,a_{\overline{j_0}}}^{r-1} &  \lambda_{j_1,a_{\overline{j_1}}}^{r-1} & \cdots & \lambda_{j_r,a_{\overline{j_r}}}^{r-1} &&&&&&&&\\
  &&&&  1 & 1 & \cdots & 1 &&&&  \\
  &&&&  \lambda_{t_0,a_{\overline{t_0}}} &  \lambda_{t_1,a_{\overline{t_1}}} & \cdots & \lambda_{t_{1},a_{\overline{t_{r}}}}&&&&\\
  &&&&  \vdots & \vdots & \ddots & \vdots&&&&\\
   &&&& \lambda_{t_0,a_{\overline{t_0}}}^{r-1} &  \lambda_{t_1,a_{\overline{t_1}}}^{r-1} & \cdots & \lambda_{t_{r},a_{\overline{t_{r}}}}^{r-1}&&&&\\
    &&&& &&&&   1 & 1 & \cdots & 1 \\
  &&&& &&&& \lambda_{u_0,a_{\overline{u_0}}} &  \lambda_{u_1,a_{\overline{u_1}}} & \cdots & \lambda_{u_{1},a_{\overline{u_{r}}}}\\
  &&&& &&&& \vdots & \vdots & \ddots & \vdots\\
   &&&&&&&& \lambda_{u_0,a_{\overline{u_0}}}^{r-1} &  \lambda_{u_1,a_{\overline{u_1}}}^{r-1} & \cdots & \lambda_{u_{r},a_{\overline{u_{r}}}}^{r-1}\\
   \lambda_{j_0,a_{\overline{j_0}}}^{r} &  \lambda_{j_1,a_{\overline{j_1}}}^{r} & \cdots & \lambda_{j_r,a_{\overline{j_r}}}^{r} &  \lambda_{t_0,a_{\overline{t_0}}}^{r} &  \lambda_{t_1,a_{\overline{t_1}}}^{r} & \cdots & \lambda_{t_{r},a_{\overline{t_{r}}}}^{r} &  \lambda_{u_0,a_{\overline{u_0}}}^{r} &  \lambda_{u_1,a_{\overline{u_1}}}^{r} & \cdots & \lambda_{u_{r},a_{\overline{u_{r}}}}^{r}\\
    \frac{\theta_i}{\lambda_{j_0,a_{\overline{j_0}}}-d_0}     &   \frac{\theta_i}{\lambda_{j_1,a_{\overline{j_1}}}-d_0} & \cdots &  \frac{\theta_i}{\lambda_{j_r,a_{\overline{j_r}}}-d_0}&  \frac{\theta_{k}}{\lambda_{t_0,a_{\overline{t_0}}}-d_0} &  \frac{\theta_{k}}{\lambda_{t_1,a_{\overline{t_1}}}-d_0} & \cdots & \frac{\theta_{k}}{\lambda_{t_{r},a_{\overline{t_{r}}}}-d_0}&  \frac{\theta_{l}}{\lambda_{u_0,a_{\overline{u_0}}}-d_0} &  \frac{\theta_{l}}{\lambda_{u_1,a_{\overline{u_1}}}-d_0} & \cdots & \frac{\theta_{l}}{\lambda_{u_{r},a_{\overline{u_{r}}}}-d_0}\\
    \frac{\delta_i}{\lambda_{j_0,a_{\overline{j_0}}}-d_1}     &   \frac{\delta_i}{\lambda_{j_1,a_{\overline{j_1}}}-d_1} & \cdots &  \frac{\delta_i}{\lambda_{j_r,a_{\overline{j_r}}}-d_1}&  \frac{\delta_{k}}{\lambda_{t_0,a_{\overline{t_0}}}-d_1} &  \frac{\delta_{k}}{\lambda_{t_1,a_{\overline{t_1}}}-d_1} & \cdots & \frac{\delta_{k}}{\lambda_{t_{r},a_{\overline{t_{r}}}}-d_1}&  \frac{\delta_{l}}{\lambda_{u_0,a_{\overline{u_0}}}-d_1} &  \frac{\delta_{l}}{\lambda_{u_1,a_{\overline{u_1}}}-d_1} & \cdots & \frac{\delta_{l}}{\lambda_{u_{r},a_{\overline{u_{r}}}}-d_1}
\end{array}
\hspace{-1mm}\right).
\end{equation}
\hrule
}
\end{figure*}

By Lemma \ref{Lemma-Det}, $\det(B_a)\ne 0$ is equivalent to the inequality in \eqref{Eqn_Ka_C3_3NF}
\begin{figure*}
{\small
\begin{align}\label{Eqn_Ka_C3_3NF}
\nonumber K_a&=\left|\begin{smallmatrix}
 \left|\begin{smallmatrix}
        1 & 1 & \cdots & 1   \\
    \lambda_{j_0,a_{\overline{j_0}}} &  \lambda_{j_1,a_{\overline{j_1}}} & \cdots & \lambda_{j_r,a_{\overline{j_r}}} \\
    \vdots & \vdots & \ddots & \vdots\\
     \lambda_{j_0,a_{\overline{j_0}}}^{r-1} &  \lambda_{j_1,a_{\overline{j_1}}}^{r-1} & \cdots & \lambda_{j_r,a_{\overline{j_r}}}^{r-1}\\
     \lambda_{j_0,a_{\overline{j_0}}}^{r} &  \lambda_{j_1,a_{\overline{j_1}}}^{r} & \cdots & \lambda_{j_r,a_{\overline{j_r}}}^{r}
 \end{smallmatrix}\right|    & \left|\begin{smallmatrix}
       1 & 1 & \cdots & 1   \\
 \lambda_{t_0,a_{\overline{t_0}}} &  \lambda_{t_1,a_{\overline{t_1}}} & \cdots & \lambda_{t_{1},a_{\overline{t_{r}}}}\\
    \vdots & \vdots & \ddots & \vdots\\
   \lambda_{t_0,a_{\overline{t_0}}}^{r-1} &  \lambda_{t_1,a_{\overline{t_1}}}^{r-1} & \cdots & \lambda_{t_{r},a_{\overline{t_{r}}}}^{r-1}\\
   \lambda_{t_0,a_{\overline{t_0}}}^{r} &  \lambda_{t_1,a_{\overline{t_1}}}^{r} & \cdots & \lambda_{t_{r},a_{\overline{t_{r}}}}^{r}
 \end{smallmatrix}\right|  & \left|\begin{smallmatrix}
     1 & 1 & \cdots & 1 \\
\lambda_{u_0,a_{\overline{u_0}}} &  \lambda_{u_1,a_{\overline{u_1}}} & \cdots & \lambda_{u_{1},a_{\overline{u_{r}}}}\\
\vdots & \vdots & \ddots & \vdots\\
 \lambda_{u_0,a_{\overline{u_0}}}^{r-1} &  \lambda_{u_1,a_{\overline{u_1}}}^{r-1} & \cdots & \lambda_{u_{r},a_{\overline{u_{r}}}}^{r-1}\\
 \lambda_{u_0,a_{\overline{u_0}}}^{r} &  \lambda_{u_1,a_{\overline{u_1}}}^{r} & \cdots & \lambda_{u_{r},a_{\overline{u_{r}}}}^{r}
 \end{smallmatrix}\right|\\
     \left|\begin{smallmatrix}
        1 & 1 & \cdots & 1   \\
    \lambda_{j_0,a_{\overline{j_0}}} &  \lambda_{j_1,a_{\overline{j_1}}} & \cdots & \lambda_{j_r,a_{\overline{j_r}}} \\
    \vdots & \vdots & \ddots & \vdots\\
     \lambda_{j_0,a_{\overline{j_0}}}^{r-1} &  \lambda_{j_1,a_{\overline{j_1}}}^{r-1} & \cdots & \lambda_{j_r,a_{\overline{j_r}}}^{r-1}\\
      \frac{\theta_i}{\lambda_{j_0,a_{\overline{j_0}}}-d_0}     &   \frac{\theta_i}{\lambda_{j_1,a_{\overline{j_1}}}-d_0} & \cdots &  \frac{\theta_i}{\lambda_{j_r,a_{\overline{j_r}}}-d_0}
 \end{smallmatrix}\right|    & \left|\begin{smallmatrix}
       1 & 1 & \cdots & 1   \\
 \lambda_{t_0,a_{\overline{t_0}}} &  \lambda_{t_1,a_{\overline{t_1}}} & \cdots & \lambda_{t_{1},a_{\overline{t_{r}}}}\\
    \vdots & \vdots & \ddots & \vdots\\
   \lambda_{t_0,a_{\overline{t_0}}}^{r-1} &  \lambda_{t_1,a_{\overline{t_1}}}^{r-1} & \cdots & \lambda_{t_{r},a_{\overline{t_{r}}}}^{r-1}\\
   \frac{\theta_{k}}{\lambda_{t_0,a_{\overline{t_0}}}-d_0} &  \frac{\theta_{k}}{\lambda_{t_1,a_{\overline{t_1}}}-d_0} & \cdots & \frac{\theta_{k}}{\lambda_{t_{r},a_{\overline{t_{r}}}}-d_0}
 \end{smallmatrix}\right|  & \left|\begin{smallmatrix}
     1 & 1 & \cdots & 1 \\
\lambda_{u_0,a_{\overline{u_0}}} &  \lambda_{u_1,a_{\overline{u_1}}} & \cdots & \lambda_{u_{1},a_{\overline{u_{r}}}}\\
\vdots & \vdots & \ddots & \vdots\\
 \lambda_{u_0,a_{\overline{u_0}}}^{r-1} &  \lambda_{u_1,a_{\overline{u_1}}}^{r-1} & \cdots & \lambda_{u_{r},a_{\overline{u_{r}}}}^{r-1}\\
 \frac{\theta_{l}}{\lambda_{u_0,a_{\overline{u_0}}}-d_0} &  \frac{\theta_{l}}{\lambda_{u_1,a_{\overline{u_1}}}-d_0} & \cdots & \frac{\theta_{l}}{\lambda_{u_{r},a_{\overline{u_{r}}}}-d_0}
 \end{smallmatrix}\right|\\
     \left|\begin{smallmatrix}
        1 & 1 & \cdots & 1   \\
    \lambda_{j_0,a_{\overline{j_0}}} &  \lambda_{j_1,a_{\overline{j_1}}} & \cdots & \lambda_{j_r,a_{\overline{j_r}}} \\
    \vdots & \vdots & \ddots & \vdots\\
     \lambda_{j_0,a_{\overline{j_0}}}^{r-1} &  \lambda_{j_1,a_{\overline{j_1}}}^{r-1} & \cdots & \lambda_{j_r,a_{\overline{j_r}}}^{r-1}\\
     \frac{\delta_i}{\lambda_{j_0,a_{\overline{j_0}}}-d_1}     &   \frac{\delta_i}{\lambda_{j_1,a_{\overline{j_1}}}-d_1} & \cdots &  \frac{\delta_i}{\lambda_{j_r,a_{\overline{j_r}}}-d_1}
 \end{smallmatrix}\right|    & \left|\begin{smallmatrix}
       1 & 1 & \cdots & 1   \\
 \lambda_{t_0,a_{\overline{t_0}}} &  \lambda_{t_1,a_{\overline{t_1}}} & \cdots & \lambda_{t_{1},a_{\overline{t_{r}}}}\\
    \vdots & \vdots & \ddots & \vdots\\
   \lambda_{t_0,a_{\overline{t_0}}}^{r-1} &  \lambda_{t_1,a_{\overline{t_1}}}^{r-1} & \cdots & \lambda_{t_{r},a_{\overline{t_{r}}}}^{r-1}\\
   \frac{\delta_{k}}{\lambda_{t_0,a_{\overline{t_0}}}-d_1} &  \frac{\delta_{k}}{\lambda_{t_1,a_{\overline{t_1}}}-d_1} & \cdots & \frac{\delta_{k}}{\lambda_{t_{r},a_{\overline{t_{r}}}}-d_1}
 \end{smallmatrix}\right|  & \left|\begin{smallmatrix}
     1 & 1 & \cdots & 1 \\
\lambda_{u_0,a_{\overline{u_0}}} &  \lambda_{u_1,a_{\overline{u_1}}} & \cdots & \lambda_{u_{1},a_{\overline{u_{r}}}}\\
\vdots & \vdots & \ddots & \vdots\\
 \lambda_{u_0,a_{\overline{u_0}}}^{r-1} &  \lambda_{u_1,a_{\overline{u_1}}}^{r-1} & \cdots & \lambda_{u_{r},a_{\overline{u_{r}}}}^{r-1}\\
 \frac{\delta_{l}}{\lambda_{u_0,a_{\overline{u_0}}}-d_1} &  \frac{\delta_{l}}{\lambda_{u_1,a_{\overline{u_1}}}-d_1} & \cdots & \frac{\delta_{l}}{\lambda_{u_{r},a_{\overline{u_{r}}}}-d_1}
 \end{smallmatrix}\right|
\end{smallmatrix}\right|\\
&\ne 0.
\end{align}}
\hrule
\end{figure*}
for all $a\in [0, \ell)$, where we also use $|A|$ to denote the determinant of the matrix $A$.
By Lemma \ref{Le_det_CV}, we further have
\begin{equation*}
K_a= D_0D_1D_2K'_a,
\end{equation*}
where 
\begin{equation*}
D_0=\prod_{0\le l<i\le r}^{r}(\lambda_{j_i,a_{\overline{j_i}}}-\lambda_{j_l,a_{\overline{j_l}}}),
\end{equation*}
\begin{equation*}
 D_1=\prod_{0\le l<i\le r}^{r}(\lambda_{t_i,a_{\overline{t_i}}}-\lambda_{t_l,a_{\overline{t_l}}}),
\end{equation*}
\begin{equation*}
D_2=\prod_{0\le l<i\le r}^{r}(\lambda_{u_i,a_{\overline{u_i}}}-\lambda_{u_l,a_{\overline{u_l}}}),
\end{equation*}
\begin{align*}
&K'_a\\=&\begin{psmallmatrix}\hspace{-2mm}
1     & 1 & 1\\
\frac{\theta_i}{\prod_{i=0}^{r}(\lambda_{j_i,a_{\overline{j_i}}}-d_0)}    & \frac{\theta_k}{\prod_{i=0}^{r}(\lambda_{t_i,a_{\overline{t_i}}}-d_0)}  & \frac{\theta_l}{\prod_{i=0}^{r}(\lambda_{u_i,a_{\overline{u_i}}}-d_0)}\\
\frac{\delta_i}{\prod_{i=0}^{r}(\lambda_{j_i,a_{\overline{j_i}}}-d_1)}    & \frac{\delta_k}{\prod_{i=0}^{r}(\lambda_{t_i,a_{\overline{t_i}}}-d_1)}  & \frac{\delta_l}{\prod_{i=0}^{r}(\lambda_{u_i,a_{\overline{u_i}}}-d_1)}
\hspace{-2mm}\end{psmallmatrix}.
\end{align*}
 By writing the Laplace expansion of the above determinant $K'_a$ over the second row, $K'_a$ can be expressed as a linear combination of $\theta_i$, $\theta_k$, and $\theta_l$.  The coefficient of $\theta_i$ is
\begin{align*}
 &\frac{1}{\prod_{i=0}^{r}(\lambda_{j_i,a_{\overline{j_i}}}-d_0)}\\
 \times &\left( \frac{\delta_k}{\prod_{i=0}^{r}(\lambda_{t_i,a_{\overline{t_i}}}-d_1)}  -\frac{\delta_l}{\prod_{i=0}^{r}(\lambda_{u_i,a_{\overline{u_i}}}-d_1)}\right),
\end{align*}
which is a  nonzero element in $\mathbf{F}_{q_0}$ because  $\frac{1}{\prod_{i=0}^{r}(\lambda_{t_i,a_{\overline{t_i}}}-d_1)}$ and $\frac{1}{\prod_{i=0}^{r}(\lambda_{u_i,a_{\overline{u_i}}}-d_1)}$ are in $G\subset \mathbf{F}_{q_0}$ while $\delta_k$ and $\delta_l$ in different cosets of $G$ in $\mathbf{F}_{q_0}\backslash\{0\}$, and $\prod_{i=0}^{r}(\lambda_{j_i,a_{\overline{j_i}}}-d_0)\in \mathbf{F}_{q_0}\backslash\{0\}$.
It can be similarly proved that the coefficients of $\theta_k$ and $\theta_l$ are also in $\mathbf{F}_{q_0}\backslash\{0\}$.
By C3, this linear combination, i.e., $K'_a$ cannot be zero. By C1, Lemma \ref{Eqn-Phi}, R1 and R2 of Lemma \ref{lem repair}, we have that $D_0D_1D_2\ne 0$. 
Therefore, $K_a\ne 0$ and $\det(B_a)\ne 0$, $\mbox{blkdiag}(B_0, B_1, \ldots, B_{\ell-1})$ and thus $\hat{\mathcal{H}}$ is nonsingular, and the original file can be reconstructed.
\end{itemize}

This completes the proof.
\end{proof}

\section{Conclusion}\label{sec:Conclu}
In this paper, we proposed two constructions of PMDS array codes with two global parities and with $(1+\epsilon)$-optimal repair bandwidth, the required finite fields, and sub-packetization levels are much smaller than the one in \cite{holzbaur2021partial}. The first one can support an arbitrary number of local parities and provide a tradeoff between the sub-packetization level and the repair bandwidth. In contrast,  the other one is limited to two local parities but has smaller rebuilding access and its sub-packetization level is only $2$. In addition, we presented an explicit PMDS array code with three global parities that has a smaller sub-packetization level as well as smaller repair bandwidth, the required finite field is also significantly smaller than the ones in \cite{holzbaur2021partial}.

\section*{Acknowledgment}
The authors would like to thank the Associate Editor
Prof. Itzhak
Tamo and the two anonymous reviewers for
their valuable suggestions and comments, which have greatly
improved the presentation and quality of this paper.

\bibliographystyle{ieeetr}
\bibliography{PMDS}
\begin{IEEEbiographynophoto}{Jie Li} (Member, IEEE) received the B.S. and M.S. degrees in mathematics from Hubei University, Wuhan, China, in 2009 and 2012, respectively, and received the Ph.D. degree from the department of communication engineering, Southwest Jiaotong University, Chengdu, China, in 2017.

From  2015 to   2016, he was a visiting Ph.D. student at the Department of Electrical Engineering and Computer Science, The University of Tennessee at Knoxville, TN, USA.  From   2017 to   2019, he was  a postdoctoral researcher at the Department of Mathematics, Hubei University, Wuhan, China. From 2019 to   2021, he was a postdoctoral researcher at the Department of Mathematics and Systems Analysis, Aalto University, Finland. He is currently a Senior Researcher with the Theory Lab, Central Research Institute, 2012 Labs, Huawei Technologies Co., Ltd., Hong Kong SAR, China. His research interests include coding for distributed storage and distributed computing, private information retrieval, and sequence design.

Dr. Li received the IEEE Jack Keil Wolf ISIT Student Paper Award and the  
Honor of Outstanding Graduates of Sichuan Province both in 2017.
\end{IEEEbiographynophoto}

\begin{IEEEbiographynophoto}{Xiaohu Tang} (Senior Member, IEEE)  received the B.S. degree in applied mathematics from
the Northwest Polytechnic University, Xi'an, China, the M.S. degree in applied
mathematics from the Sichuan University, Chengdu, China, and the Ph.D.
degree in electronic engineering from the Southwest Jiaotong University,
Chengdu, China, in 1992, 1995, and 2001 respectively.

From 2003 to 2004, he was a research associate in the Department of Electrical
and Electronic Engineering, Hong Kong University of Science and Technology.
From 2007 to 2008, he was a visiting professor at University of Ulm,
Germany. Since 2001, he has been in the School of Information Science and Technology,
Southwest Jiaotong University, where he is currently a professor. His research
interests include coding theory, network security, distributed storage and information processing for big data.

Dr. Tang was the recipient of the National excellent Doctoral Dissertation
award in 2003 (China), the Humboldt Research Fellowship in 2007
(Germany), and the Outstanding Young Scientist Award by NSFC in 2013
(China). He served as Associate Editors for several journals including \textit{IEEE
Transactions on Information Theory} and \textit{IEICE Transactions on
Fundamentals}, and served on a number of technical program committees of
conferences.
\end{IEEEbiographynophoto}

\begin{IEEEbiographynophoto}{Hanxu Hou} (Member, IEEE) received the B.Eng. degree in Information Security from Xidian University, Xian, China, in 2010, and Ph.D. degrees in the Dept. of Information Engineering from The Chinese University of Hong Kong in 2015 and in the School of Electronic and Computer Engineering from Peking University in 2016. He is currently with Dongguan University of Technology, and Part-time Researcher of the Theory Lab, Central Research Institute, 2012 Labs, Huawei Technologies Co., Ltd. He was a recipient of the 2020 Chinese Information Theory Young Rising Star Award by China Information Theory Society. He was recognized as an Exemplary Reviewer 2020 in IEEE Transactions on Communications. His research interests include erasure coding and coding for distributed storage systems. 
\end{IEEEbiographynophoto}

\begin{IEEEbiographynophoto}{Yunghsiang S.~Han} (Fellow, IEEE) was born in Taipei, Taiwan, in 1962. He received B.Sc. and M.Sc. degrees in electrical engineering from the National Tsing Hua University, Hsinchu, Taiwan, in 1984 and 1986, respectively, and a Ph.D. degree from the School of Computer and Information Science, Syracuse University, Syracuse, NY, in 1993. From 1986 to 1988, he was a lecturer at Ming-Hsin Engineering College, Hsinchu, Taiwan. He was a teaching assistant from 1989 to 1992 and a research associate in the School of Computer and Information Science at Syracuse University from 1992 to 1993. From 1993 to 1997, he was an Associate Professor in the Department of Electronic Engineering at Hua Fan College of Humanities and Technology, Taipei Hsien, Taiwan. He was with the Department of Computer Science and Information Engineering at National Chi Nan University, Nantou, Taiwan, from 1997 to 2004. He was promoted to Professor in 1998. He was a visiting scholar in the Department of Electrical Engineering at the University of Hawaii at Manoa, HI, from June to October 2001, the SUPRIA visiting research scholar in the Department of Electrical Engineering and Computer Science and CASE center at Syracuse University, NY from September 2002 to January 2004 and July 2012 to June 2013, and the visiting scholar in the Department of Electrical and Computer Engineering at University of Texas at Austin, TX from August 2008 to June 2009. He was with the Graduate Institute of Communication Engineering at National Taipei University, Taipei, Taiwan, from August 2004 to July 2010. From August 2010 to January 2017, he was with the Department of Electrical Engineering at the National Taiwan University of Science and Technology as Chair Professor. From February 2017 to February 2021, he was with the School of Electrical Engineering \& Intelligentization at Dongguan University of Technology, China. Now he is with the Shenzhen Institute for Advanced Study, University of Electronic Science and Technology of China. He is also a Chair Professor at National Taipei University since February 2015. His research interests are in error-control coding, wireless networks, and security.

Dr. Han was a winner of the 1994 Syracuse University Doctoral Prize and a Fellow of IEEE. One of his papers won the prestigious 2013 ACM CCS Test-of-Time Award in cybersecurity.

\end{IEEEbiographynophoto}

\begin{IEEEbiographynophoto}{Bo Bai} (Senior Member, IEEE) received the B.S. 
degree (Hons.) from the School of Communication Engineering, Xidian University, Xi’an, China, 
in 2004, and the Ph.D. degree from the Department
of Electronic Engineering, Tsinghua University, Beijing, China, in 2010. 

He was a Research Assistant and a Research 
Associate with the Department of Electronic and 
Computer Engineering, The Hong Kong University of Science and Technology, from April 2009 to 
September 2010 and October 2010 to April 2012,  
respectively. From July 2012 to January 2017, he was an Assistant Professor  
with the Department of Electronic Engineering, Tsinghua University. He has  
obtained the support from the Backbone Talents Supporting Project of  
Tsinghua University. Currently, he is an Information Theory Scientist and the  
Director of the Theory Laboratory, 2012 Labs, Huawei Technologies Company  
Ltd., Hong Kong. He has authored more than 130 papers in major IEEE and  
ACM journals and conferences, two book chapters, and one textbook. His 
research interests include semantic information theory, B5G/6G networking,  
and graph informatics. He served as a member for the IEEE ComSoc WTC  
Committee and the IEEE ComSoc SPCE Committee. He also served as a  
TPC Member for several IEEE conferences, such as ICC, Globecom, WCNC,  
VTC, and ICCC. He was a recipient of the Student Travel Grant at IEEE  
Globecom 2009 and the Best Paper Award at the IEEE ICC 2016. He was  
invited as a Young Scientist Speaker at the IEEE TTM 2011. He received the  
Honor of Outstanding Graduates of Shaanxi Province and the Honor of Young 
Academic Talent of Electronic Engineering in Tsinghua University. He is one  
of the Founding Vice Chair of the IEEE TCCN SIG on Social Behavior  
Driven Cognitive Radio Networks. He served as the TPC Co-Chair for the 
IEEE Infocom 2018—1st AoI Workshop and the IEEE Infocom 2019—2nd  
AoI Workshop, the TPC Co-Chair for the IEEE ICCC 2018, and an Industrial  
Forum and the Exhibition Co-Chair for the IEEE HotICN 2018. He also served  
as a reviewer for several major IEEE and ACM journals and conferences.
\end{IEEEbiographynophoto}

\begin{IEEEbiographynophoto}{Gong Zhang} is the Principal Researcher of future
network architecture with Huawei 2012 Labs. He has
more than 18 years of research experience on system
architect, including networking, distributed systems,
and communication systems. He has contributed
more than 90 patents globally. Previously as a Senior
Researcher, he led future internet and cooperative
communication research, and did mobility research
program since 2005. In 2009, he became the Principal Researcher, being in charge of Advance Network
Technology Research Department. He led lots of
research directions like future networks, distributed computing, database
systems, and data analysis. He proposed stream-based research systems in
networks to maintain the networks and find added value for the carriers.
Since 2012, he has been the Principal Researcher leading the System Group in
data mining and machine learning. His major research directions are network
architecture and large scale distributed systems.
\end{IEEEbiographynophoto}

\end{document}